\documentclass[11p, a4paper]{article}
\usepackage{breakurl}           
\usepackage{graphicx}
\usepackage{lipsum}
\usepackage{subfig}
\usepackage{amsmath, amssymb, amsthm, amscd, mathtools}
\usepackage{tikz}
\usepackage{tikz-cd}
\usetikzlibrary{arrows}
\usetikzlibrary{calc}
\usepackage{xcolor}
\usepackage[toc,page]{appendix}
\usepackage{graphicx}
\usepackage{hyperref}
\usepackage{authblk}
\usepackage{multicol}

\usepackage{mathrsfs}

\theoremstyle{definition}
\newtheorem{defn}{Definition}[section]

\theoremstyle{plain}
\newtheorem{theorem}[defn]{Theorem}
\newtheorem{conjecture}[defn]{Conjecture}
\newtheorem{corollary}[defn]{Corollary}
\newtheorem{lemma}[defn]{Lemma}
\newtheorem{proposition}[defn]{Proposition}

\theoremstyle{remark}
\newtheorem{rem}[defn]{Remark}

\theoremstyle{definition}

\newenvironment{claim}[2]{\par\noindent\textbf{Claim #1.}\space#2}{}
\newenvironment{claimproof}[1]{\par\noindent\textit{Proof}\space#1}{\hfill $\blacksquare$}

\usepackage{blindtext}

\newcommand{\tuple}[1]{\left\langle#1\right\rangle}
\newcommand{\setdef}  [2]{\mathopen{}\left\{#1 \mathrel{\Bigg|} #2\right\}\mathclose{}} 

\newcommand{\LC}{\mathsf{LC}}
\newcommand{\LCs}[2][\S]{\LC\left(#1,#2\right)}
\newcommand{\SC}{\mathsf{SC}}
\newcommand{\SCs}[1][\S]{\SC\left(#1\right)}
\newcommand{\CLC}{\mathsf{CLC}}
\newcommand{\CLCs}[2][\S]{\CLC\left(#1,#2\right)}
\newcommand{\CSC}{\mathsf{CSC}}
\newcommand{\CSCs}[1][\S]{\CSC\left(#1\right)}

\newcommand{\LCk}[1]{\LC^{(#1)}}

\newcommand{\CLCk}[1]{\CLC^{(#1)}}

\newcommand{\scenario}{\tuple{X, \M, (O_m)_m}}

\newcommand{\Meas}[1]{X^{(#1)}}
\newcommand{\Cov}[1]{\M^{(#1)}}
\newcommand{\Event}[1]{\E^{(#1)}}
\newcommand{\Model}[1]{\S^{(#1)}}
\newcommand{\FModel}[1]{\F^{(#1)}}

\newcommand{\Jscenario}[1]{\tuple{\Meas{#1},\Cov{#1}, (O^{(#1)}_m)_{m\in X^{(#1)}}}}

\newcommand{\flatten}{\mathsf{flatten}}

\newcommand{\C}{\mathscr{C}}
\newcommand{\D}{\mathcal{D}}
\newcommand{\E}{\mathcal{E}}
\newcommand{\F}{\mathcal{F}}

\newcommand{\K}{\mathcal{K}}

\newcommand{\M}{\mathcal{M}}
\newcommand{\N}{\mathcal{N}}

\renewcommand{\P}{\mathcal{P}}

\renewcommand{\S}{\mathcal{S}}

\newcommand{\U}{\mathcal{U}}

\allowdisplaybreaks

\title{Towards a complete cohomological invariant for non-locality and contextuality}
\author{Giovanni Car\`{u}}
\affil{Department of Computer Science\\ University of Oxford}
\date{}

\begin{document}
\maketitle
\begin{abstract}
The sheaf theoretic description of non-locality and contextuality by Abramsky and Brandenburger sets the ground for a topological study of these peculiar features of quantum mechanics. This viewpoint has been recently developed thanks to sheaf cohomology, which provides a sufficient condition for contextuality of empirical models in quantum mechanics and beyond. Subsequently, a number of studies proposed methods to detect contextuality based on different cohomology theories. However, none of these cohomological descriptions succeeds in giving a full invariant for contextuality applicable to concrete examples. In the present work, we introduce a cohomology invariant for possibilistic and strong contextuality which is applicable to the vast majority of empirical models.  

\end{abstract}

\section{Introduction}
Non-locality and contextuality are key features of quantum mechanics, which have been proved to play a crucial role as a fundamental resource for quantum information and computation \cite{Howard, Raussendorf4}. Abramsky and Brandenburger gave a general and unified description of these phenomena using sheaf theory \cite{Abramsky1}, showing that contextual behaviour can be observed even beyond quantum mechanics. The sheaf theoretic framework provides a rigorous topological description of contextuality, which perfectly conveys the idea of contextuality as a fundamental discrepancy between local consistency and global inconsistency \cite{Abramsky2}. In this framework, empirical models, which contain all the information concerning the outcomes of an ideal experiment, are represented as presheafs over the set of available measurements. Then, contextuality corresponds to the impossibility of extending the local sections of the empirical model presheaf to global ones. 

In recent work, this topological viewpoint has been further developed by taking advantage of sheaf cohomology, a widely used theory in algebraic geometry and topology, suited to study extendability of local features to global ones. In particular, Abramsky et al.~used \v{C}ech cohomology to derive a cohomological obstruction to the extendability of local sections \cite{Abramsky3, Abramsky2}. Although the obstruction has been proved to detect contextuality in a variety of well-studied empirical models, such as the Popescu-R\"{o}hrlich (PR) boxes \cite{Peres}, the Greenberger-Holt-Zeilinger (GHZ) states under Pauli measurements \cite{Greenberger, Greenberger2, Mermin2}, the Peres-Mermin `magic' square, and the whole class of models admitting All-vs-Nothing (AvN) arguments \cite{Abramsky2}, there is evidence of a significant amount of false positives (e.g. the Hardy model \cite{Hardy}). In fact, subsequent work has highlighted many limits of the \v{C}ech cohomology approach, including the fact that it does not provide a full invariant for contextuality even in its strongest form \cite{Caru}, and even under strong assumptions on the measurement scenarios.

Since then, other studies have used different cohomology theories to study contextual features. For instance, 
Okay et al.~\cite{RaussendorfCohomology} used simplicial and group cohomology to present topological counterparts for the contextuality arguments used in measurement based quantum computation (MBQC). However, these methods do not represent a full invariant for contextuality, and they are limited to the class of contextuality arguments relevant for MBQC. Roumen's \emph{cohomology of effect algebras} \cite{Roumen}, provides an alternative viewpoint, based on cyclic and order cohomology. The author partially addresses the question of existence of false positives, and suggests that order cohomology does not produce any false positives. However, 
the result is obtained at the expense of the practical computatability of the cohomology invariant, which appears so complex that no application to concrete empirical models has been presented yet.

The aim of this paper is to refine the theory of \v{C}ech cohomology in order to obtain a complete invariant which is applicable to the vast majority of empirical models, without compromising the practical computability of the cohomology obstructions. 

\section{Background}\label{sec: background}
\subsection{Sheaf theory and contextuality}\label{sec: basics}
The sheaf-theoretic description of non-locality and contextuality provides the appropriate topological framework for the definition of cohomology. The key aspect of this approach is that sheaves represent the most natural objects to study the extendability of local properties to global ones, and thus perfectly convey the conception of contextuality as a discrepancy between local consistency and global inconsistency. In this section, we review the main definitions and results and set the ground for the definition of cohomology. 

Let $X$ be a finite set of measurement labels representing all the measurements available to the experimenters in an ideal scenario. A fundamental aspect of contextuality scenarios is that not all the measurements may be performed simultaneously.\footnote{In quantum mechanics, this arises e.g. when considering non-commuting observables.} To capture this feature, we introduce a \emph{measurement cover} i.e. an antichain $\M\subseteq \mathcal{P}(X)$ satisfying $\bigcup_{C\in\M}C=X$. Elements of the cover $\M$ are the \emph{measurement contexts}, i.e. the maximal sets of measurements that can be jointly performed. Each measurement $m\in X$ will produce an outcome in a set $O_m$. The set $X$, the cover $\M$ and the outcome sets $(O_m)_{m\in X}$ constitute the \emph{measurement scenario} $\tuple{X,\M,(O_m)_m}$. 
We equip $X$ with the discrete topology and define the \emph{sheaf of events}
\[
\mathcal{E}:\textbf{Open}(X)^{op}=\mathcal{P}(X)^{op}\longrightarrow \textbf{Set},
\]
where, for all $U\subseteq X$,
\[
\mathcal{E}(U):=\prod_{m\in U}O_m,
\]
and the restriction maps are given by the obvious projection, i.e. given $U\subseteq U'\subseteq X$, we have
\[
\rho_U^{U'}:=\mathcal{E}(U\subseteq U'): \prod_{m\in U'}O_m\longrightarrow \prod_{m\in U}O_m::\tuple{s_m}_{m\in U'}\longmapsto \tuple{s_m}_{m\in U}.
\]
It is quite easy to show that $\mathcal{E}$ is indeed a functor and satisfies the sheaf condition.
Each $s\in \mathcal{E}(U)$ is called a \emph{local section over $U$}. If $U=X$ such a section is called a \emph{global section}. We will often refer to a local section at $U$ as a function $U\rightarrow\coprod_{m\in U} O_m$ such that $s(m)\in O_m$ for all $m\in U$. 

The measurement scenario and the sheaf of events define the experiment setting and are therefore independent of any physical system we aim to apply these measurements to. The application of this scenario on an actual physical system is captured by the notion of \emph{empirical model}. An empirical model is a compatible family $\{e_C\}_{C\in\M}$ of probability distributions over the events at each contexts $\mathcal{E}(C)$,\footnote{Here, compatibility means that the marginals of the distributions agree on the intersections of contexts. The formal definition can be found in \cite{Abramsky1} and will not be needed in this paper.} which represent the statistics obtained as a result of the experiment. In this work, we focus on \emph{possibilistic empirical models}, in other words, we shall be concerned only with whether an event is possible (i.e. with probability $>0$) or not, disregarding the actual value of the probability. 
Such models are defined as subpresheaves $\S$ of $\mathcal{E}$ satisfying the following conditions:
\begin{enumerate}
\item $\S(C)\neq \emptyset$ for all $C\in\M$.\label{cond: 1}
\item \label{cond: flasque} $\S$ is \emph{flasque beneath the cover}, i.e. the restriction map $\rho_U^{U'}$ is surjective whenever $U\subseteq U'\subseteq C$ for some context $C\in\M$. 
\item Every family $\{s_C\in\S(C)\}_{C\in\M}$ which is \emph{compatible} (i.e. for all $C,C'\in\M$, we have $s_C\mid_{C\cap C'}=s_{C'}\mid_{C\cap C'}$) induces a global section in $\S(X)$. Such a section must be unique as $\S$ is a subpresheaf of $\mathcal{E}$. \label{cond: 3}
\end{enumerate}
Condition \ref{cond: flasque} can be seen as a possibilistic version of \emph{no-signalling}. 

We can now define the notion of contextuality. Let $\S$ be an empirical model on a scenario $\scenario$. 
\begin{itemize}
\item Given a context $C\in\M$ and a section $s\in\S(C)$, we say that $\S$ is \emph{possibilistically} (or \emph{logically}) \emph{contextual} at $s$ and write $\mathsf{LC}(\S, s)$, if $s$ is not part of any compatible family. We say that $\S$ is \emph{logically contextual}, or $\mathsf{LC}(\S)$ if $\mathsf{LC}(\S, s)$ for some section $s$. 
\item We say that $\S$ is \emph{strongly contextual}, and write $\mathsf{SC}(\S)$, if $\mathsf{LC}(\S, s)$ for all $s$. In other words, by condition \ref{cond: 3}, there is no global section, i.e. $\S(X)=\emptyset$. 
\end{itemize}

In this framework, non-locality is a special case of contextuality. It corresponds to contextuality in the particular case where the measurement scenario is \emph{Bell-type}. A scenario $\scenario$ is said to be Bell-type if
\begin{itemize}
\item The measurement set $X$ can be partitioned into subsets $\{X_i\}_{i\in I}$, where $I$ labels different `parts' of the system, and $X_i$ represents the measurements that can be carried out at part $i$.
\item The cover $\M$ contains contexts of the form $\{x_i\}_{i\in I}$, where $x_i\in X_i$ for all $i\in I$. This corresponds to performing one and only one measurement for each part of the system.
\end{itemize}

For this reason, every result concerning contextuality will be true in particular for non-locality.

\subsubsection{Simplicial complex description and bundle diagrams}
The structure of a measurement scenario can also be described as an abstract simplicial complex having measurements as vertices \cite{Barbosa, Rui}. A set of vertices constitutes a face of the complex whenever the corresponding measurements can be performed jointly, hence contexts correspond to facets. Using this description, it is possible to see possibilistic empirical models as simplicial bundles over measurement complexes.
This allows to represent simple empirical models in a clear and intuitive way, using \emph{bundle diagrams}. For example, consider the Hardy model represented in Table \ref{tab: Hardy}.
\begin{table}[htbp]
\centering
\begin{tabular}{c c | c c c c}
\hline
$A$ & $B$ & $(0,0)$ & $(1,0)$ & $(0,1)$ & $(1,1)$\\
\hline
$a_1$ & $b_1$ & $1$ & $1$ & $1$ & $1$\\
$a_1$ & $b_2$ & $0$ &  $1$ &  $1$ &  $1$\\
$a_2$ & $b_1$ &  $0$ &  $1$ &  $1$ &  $1$\\
$a_2$ & $b_2$ &  $1$ &  $1$ &  $1$ &  $0$\\
\end{tabular}
\caption{The Hardy possibilistic model.}\label{tab: Hardy}
\end{table}
In this case, we have two experimenters Alice and Bob who can choose between two dichotomic measurements each ($a_1,a_2$ for Alice and $b_1,b_2$ for Bob). Thus, we have $X=\{a_1,a_2,b_1,b_2\}$, $\M=\{\{a_1,b_1\},\{a_1,b_2\},\{a_2,b_1\},\{a_2,b_2\}\}$ and $O_m=\{0,1\}$ for all $m\in X$. Notice that this is a Bell-type scenario. The rows of the table correspond to the contexts, and the events marked with '$1$' are the ones deemed possible by the model. 
The bundle description of the Hardy model can be found in Figure \ref{fig: Hardy}, both in its 3-dimensional and planar version.
\begin{figure}[htbp]
\centering
\includegraphics[scale=0.6]{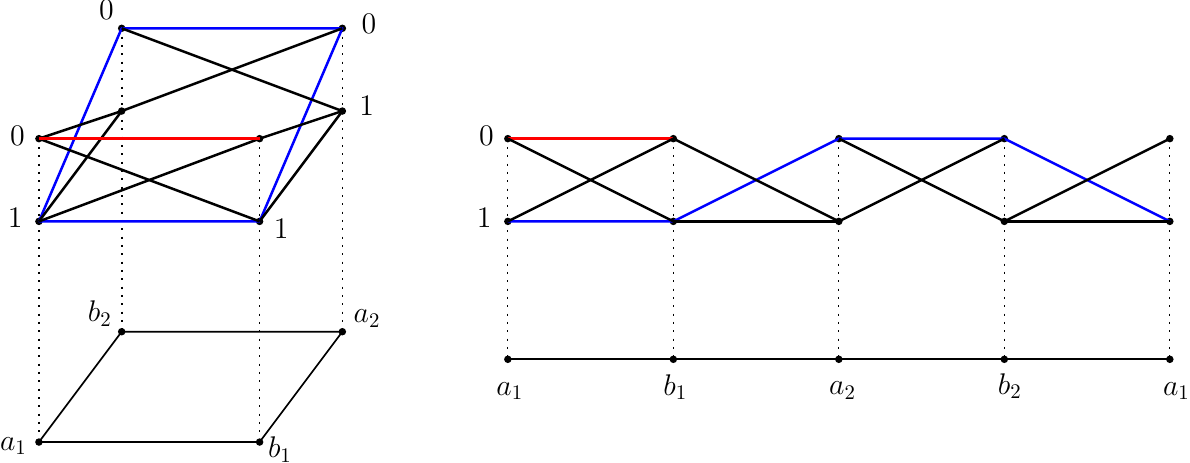}
\caption{Bundle diagram of the Hardy model in its 3-dimensional version (left) and planar version (right).}\label{fig: Hardy}
\end{figure}
At the base of the diagram lies the measurement complex. Above each vertex is a fiber representing the two possible outcomes for the corresponding measurement. Possible sections of the empirical model are represented by edges connecting points in the fiber above each context. Global sections correspond to a choice of one section per context such that they all agree at intersections, and appear as closed loops around the bundle. For instance, we have highlighted a global section in blue. The planar representation will be handy with more complicated scenarios. To recover the full model from it, one has to glue the right hand side with the left hand side. 

One of the advantages of representing empirical models using bundle diagrams is that we have an immediate visual feedback on the contextual properties of the model. For example, by simply looking at the diagrams in Figure \ref{fig: Hardy}, we see that the section $s:=(a_1,b_1)\mapsto (0,0)$ is not part of any compatible family, hence the model is logically contextual at $s$. However, the model is not strongly contextual because it contains, among others, the global section highlighted in blue.

\subsection{\v{C}ech cohomology}
In this section, we recall the main definitions and results of \cite{Abramsky3, Abramsky2} on how to apply \v{C}ech cohomology, which is a particular type of sheaf cohomology, to the study of contextuality.

Consider an empirical model $S$ on a measurement scenario $\scenario$. The first issue we encounter in the definition of cohomology is that it requires a presheaf of abelian groups, while $\S$ is merely a presheaf of sets. This difficulty is overcome by taking a presheaf $\F$ of abelian groups which \emph{represents} $\S$. Formally, this means that $\F$ is a presheaf $\F:\P(X)\rightarrow\textbf{AbGrp}$ which verifies conditions \ref{cond: 1}, \ref{cond: flasque}, \ref{cond: 3}
and is such that there exists an injection $i:\S\hookrightarrow \F$ with $i_C(s_C)\neq 0\in\F(C)$ for all $C\in \M$ and for each $s_C\in\S(C)$. In practice, we will select $\F:=F_R\S$ as a representative presheaf, where $R$ is a ring, 
and $F_R:\textbf{Set}\rightarrow \textbf{AbGrp}$ is the functor that maps a set $X$ to the free abelian group on $R$ generated by it.\footnote{The restriction maps in this case are obtained by linearly extending the ones of $\S$. In a slight abuse of notation, we will denote by $\rho_U^{U'}$ both the restriction maps of $\S$ and those of $\F$.} Although this might seem as a minor alteration, it actually plays a crucial role in the existence of false positives in the detection of contextuality as we shall see in detail later in this paper. 

The \emph{nerve} $\N(\M)$ of $\M$ is an abstract simplicial complex with vertices in $\M$. The set $\N(\M)^q$ of its $q$-simplices is constituted by tuples $\sigma=\tuple{C_0,\dots, C_q}$ of elements of $\M$ such that $|\sigma|:=\cap_{i=0}^qC_i\neq\emptyset$. For all $q\ge0$ and each $0\leq j\leq q$, we define the boundary maps $\partial_j:\mathcal{N}(\M)^{q+1}\rightarrow\mathcal{N}(\M)^q$ by
\[
\partial_j(C_0,\dots, C_{q+1}):=(C_0,\dots, C_{j-1}, \hat{C_j}, C_{j+1},\dots, C_{q+1}).
\]
We can now introduce the \emph{augmented \v{C}ech cochain complex}
\[
0\xrightarrow{~0~} C^0(\M, \F)\xrightarrow{\delta^0}C^1(\M, \F)\xrightarrow{\delta^1}\dots
\]
where, for all $q\ge 0$,
\[
C^q(\M,\F):=\bigoplus_{\sigma\in\mathcal{N}(\M)^q}\F(|\sigma|)
\]
is the abelian group of $q$-\emph{cochains}, and $\delta^q:C^q(\M,\F)\rightarrow C^{q+1}(\M,\F)$ defined by
\[
\delta^q(\omega)(\sigma):=\sum_{j=0}^{q+1}(-1)^j\rho_{|\sigma|}^{|\partial_j\sigma|}(\omega(\partial_j\sigma))~~\forall\omega\in C^q(\M,\F),~\forall\sigma\in\mathcal{N}(\M)^q
\]
is the $q$-\emph{th coboundary map}. \emph{\v{C}ech cohomology} $\check{H}^\ast(\M,\F)$ is defined as the cohomology of this augmented cochain complex.

We will assume that $\M$ is \emph{connected}, which means that given any $C,C'\in\M$, there exists a sequence $C=C_0,\dots, C_n=C'$ such that $C_i\cap C_{i+1}\neq\emptyset$ for all $0\leq i\leq n-1$. Note that this assumption does not cause any loss of generality as we can always study an empirical model defined on a non-connected cover by 
analysing its behavior on the individual connected components, without compromising the contextual structure of the model itself. With this assumption, cocycles in $Z^0(\M,\F)\cong\check{H}^0(\M,\F)$ correspond to compatible families $\{r_C\in\F(C)\}_{C\in\M}$, i.e. families verifying $r_C\mid_{C\cap C'}=r_{C'}\mid_{C\cap C'}$ for all $C,C'\in\M$.\footnote{Here, $r_C\mid_{C\cap C'}$ is an equivalent notation for $\rho^C_{C\cap C'}(r_C)=\F(C\cap C'\subseteq C)(r_C)$.} 

We shall be concerned with extendability of local sections at a fixed context $C_0\in\M$. For this reason, we define the \emph{relative cohomology of} $\F$. To do so, we introduce two auxiliary preshaves. Firstly
\[
\F\mid_{C_0}:\textbf{Open}(X)^{op}\rightarrow\textbf{AbGrp}::U\mapsto \F(U\cap C_0).
\]
The restriction to $C_0$ yields a morphism of sheaves $p^{C_0}:\F\Rightarrow\F\mid_{C_0}$ given by
\[
p^{C_0}_U:\F(U)\rightarrow\F\mid_{C_0}(U)::r\mapsto r\mid_{C_0\cap U}.
\]
Each $p^{C_0}_U$ is surjective as $\F$ is flasque beneath the cover and $U\cap C_0\subseteq C_0\in\M$. 
The second presheaf is defined by $\F_{\tilde{C}_0}(U):=\ker(p_U^{C_0})$. To summarise, we have the following exact sequence of presheaves
\begin{equation}\label{equ: exact sequence}
\textbf{0}\Longrightarrow\F_{\tilde{C}_0}\Longrightarrow\F\xRightarrow{~p^{C_0}}\F\mid_{C_0},
\end{equation}
which can be lifted to cochains to
\[
0\longrightarrow C^0(\M, \F_{\tilde{C}_0})\xhookrightarrow{~~~~~~~}C^0(\M, \F)\xrightarrow{\bigoplus_Cp_C^{C_0}}C^0(\F, \F\mid_{C_0}),\longrightarrow 0,
\]
where exactness on the right follows by surjectivity of all the $p_C^{C_0}$. The map $\delta^0$ can be correstricted to a map $\tilde{\delta^0}:=\delta^0\mid^{Z^1(\M,\F)}$ whose kernel is $Z^0(\M,\F)\cong\check{H}^0(\M,\F)$ and whose cokernel is isomorphic to $\check{H}^1(\M,\F)$, and the same procedure can be applied to $\F\mid_{C_0}$ and $\F_{\tilde{C}_0}$. Therefore, by applying the snake lemma to 
\begin{center}
\begin{tikzpicture}[auto]
\matrix (m) [matrix of math nodes, row sep=2em, column sep=2em, text height=2ex, text depth=0.25ex]
{
0 & C^0(\M, \F_{\tilde{C}_0}) & C^0(\M, \F) & C^0(\M, \F\mid_{C_0}) & 0 \\
0 & Z^1(\M, \F_{\tilde{C}_0}) & Z^1(\M, \F) & Z^1(\M, \F\mid_{C_0}) & \\
};

\path[->, font=\scriptsize]
(m-1-1) edge node[auto] {} (m-1-2)
(m-1-2) edge node[auto] {} (m-1-3)
(m-1-3) edge node[auto] {} (m-1-4)
(m-1-4) edge node[auto] {} (m-1-5)
(m-2-1) edge node[auto] {} (m-2-2)
(m-2-2) edge node[auto] {} (m-2-3)
(m-2-3) edge node[auto] {} (m-2-4)
(m-1-2) edge node[auto] {$\tilde{\delta}^0$} (m-2-2)
(m-1-3) edge node[auto] {$\tilde{\delta}^0$} (m-2-3)
(m-1-4) edge node[auto] {$\tilde{\delta}^0$} (m-2-4);

\end{tikzpicture}
\end{center}
we obtain the following exact sequence

 \begin{center}
\begin{tikzpicture}[descr/.style={fill=white,inner sep=1.5pt}]
        \matrix (m) [
            matrix of math nodes,
            row sep=2.5em,
            column sep=2.5em,
            text height=1.5ex, text depth=0.25ex
        ]
        { \check{H}^0(\M,\F_{\tilde{C}_0}) & \check{H}^0(\M,\F) & \check{H}^0(\M,\F\mid_{C_0}) \\
        \check{H}^1(\M,\F_{\tilde{C}_0}) & \check{H}^1(\M,\F) & \check{H}^1(\M,\F\mid_{C_0})\\
        };

        \path[overlay,->, font=\scriptsize,>=latex]
        (m-1-1) edge node[auto] {} (m-1-2)
        (m-1-2) edge node[auto] {} (m-1-3)
        (m-1-3) edge[out=355,in=175] node[descr,yshift=0.3ex] {$\gamma_{C_0}$} (m-2-1)
        (m-2-1) edge node[auto] {} (m-2-2)
        (m-2-2) edge node[auto] {} (m-2-3);
\end{tikzpicture}
\end{center}
The homomorphism $\gamma_{C_0}$ is called the \emph{connecting homomorphism} relative to the context $C_0$.

It can be shown that $\F(C_0)\cong\check{H}^0(\M, \F\mid_{C_0})$. Thus, given a local section $r_0\in\F(C_0)$, we can define the \emph{cohomology obstruction} of $r_0$ as the element $\gamma_{C_0}(r_0)\in\check{H}^1(\M,\F_{\tilde{C}_0})$.


\begin{proposition}[\cite{Abramsky3}]\label{prop: main}
Let $\M$ be a connected cover, $C_0\in\M$ and $r_0\in\F(C_0)$. Then, $\gamma_{C_0}(r_0)=0$ if and only if there exists a compatible family $\{r_C\in\F(C)\}_{C\in\M}$ such that $r_{C_0}=r_0$. 
\end{proposition}


Let $\S$ be an empirical model and consider a local section $s_0\in\S(C_0)$. We define the following concepts
\begin{itemize}
\item $\S$ is \emph{cohomologically logically contextual at $s_0$}, or $\CLCs{s_0}$, if $\gamma_{C_0}(s_0)\neq 0$. We say that $\S$ is \emph{cohomologically logically contextual}, or $\CLC(\S)$, if $\CLCs{s}$ for some section $s$. 
\item $\S$ is \emph{cohomologically strongly contextual}, or $\CSCs$, if $\CLCs{s}$ for all sections $s$.
\end{itemize}

The main result of \cite{Abramsky3} provides a sufficient condition for an empirical model to be contextual:
\begin{theorem}\label{thm: main}
Let $\S$ be an empirical model. Given a section $s_0$ of $\S$, we have $\CLC(\S,s_0)\Rightarrow \LC(\S,s_0)$. Moreover, $
\CSCs\Rightarrow \SCs$.
\end{theorem} 

Note that cohomology only provides a sufficient condition for contextuality, which is not necessary in general. False positives do exist, as we shall see in detail in the following section. The key subtlety that leads to such detection errors is that the vanishing of the cohomology obstruction $\gamma(s)$ implies that $s$ is a part of a compatible family for the presheaf $\F$. However \emph{this compatible family may not be a valid family for $\S$}, because $\S$ does not allow linear combinations of sections.

 \section{False positives in cohomology}\label{sec: false positives}

In this section, we will study false positives by analysing some key examples.

Consider the Hardy model presented in Figure \ref{fig: Hardy}. We have already shown that the section $s:=(a_1, b_1)\mapsto (0,0)$, marked in red in the picture, is not part of any compatible family of the presheaf $\S$ of the model, proving that $\S$ is logically contextual at $s$. However, the section $s$ is part of the following compatible family for the sheaf $\F:= F_\mathbb{Z}\S$:
\[
\begin{split}
\{s, (a_2,b_1)\mapsto (1,0), & (a_2, b_2)\mapsto(1,0), \\
& [(a_1,b_2)\mapsto (1,0)]-[(a_1,b_2)\mapsto (1,1)]+[(a_1,b_2)\mapsto (0,1)]\},
\end{split}
\]
which is higlighted in blue in the bundle diagram of Figure \ref{fig: Hardy counter}.\footnote{The compatible family can be seen as a closed loop around the bundle, or, equivalently, as a path going from left to right in the planar bundle diagram. The convention we follow is that we assign coefficient $+1$ to all the segments of the path going from left to right, and $-1$ to all the segments going from right to left. This kind of loops will be referred to as \emph{cohomology loops} or \emph{non-standard loops}, and correspond to global sections for the presheaf $\F$.}
\begin{figure}[htbp]
\centering
\includegraphics[scale=0.6]{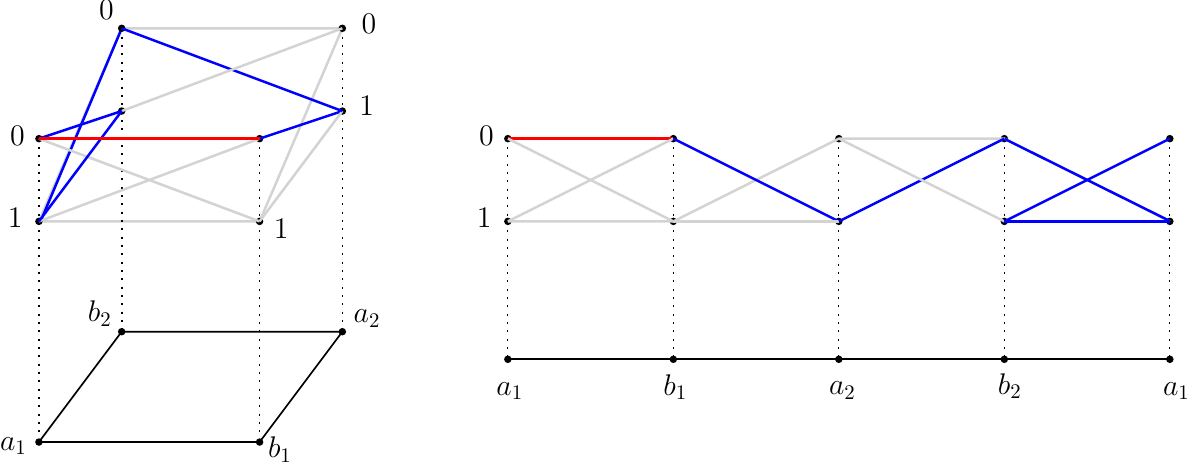}
\caption{A cohomology false positive for the Hardy model.}\label{fig: Hardy counter}
\end{figure}
This means that the obstruction $\gamma(s)$ vanishes, and cohomology is unable to detect the contextuality of the Hardy model. 

Another, more extreme, false positive is presented in \cite{Caru}, and is depicted in the bundle diagram of Figure \ref{fig: Gio}.
\begin{figure}[htbp]
\centering
\includegraphics[scale=0.6]{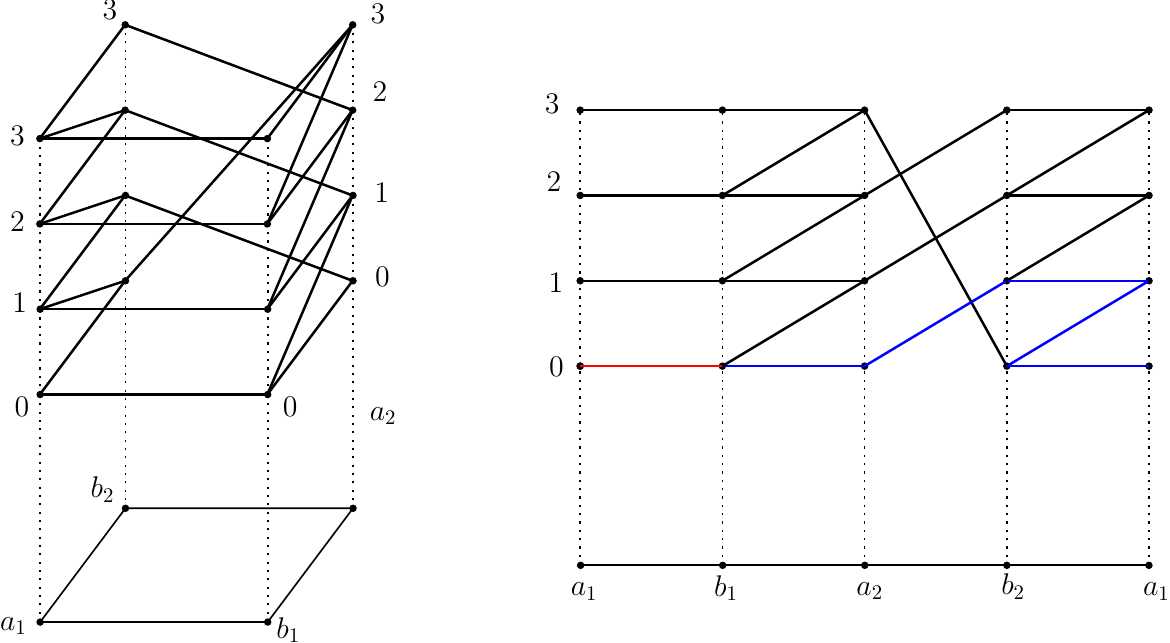}
\caption{A strongly contextual model which is cohomologically non-contextual.}\label{fig: Gio}
\end{figure}
The peculiarity of this model is that cohomology fails to detect contextuality at every single section of the model, despite the fact that none of the sections is part of a compatible family, as the model is strongly contextual. This proves that \v{C}ech cohomology, as it is defined in \cite{Abramsky3}, is not an invariant for strong contextuality, and it can fail even on extremely simple scenarios. To show how this is possible, in the planar diagram of Figure \ref{fig: Gio}, we highlighted in blue a compatible family for $\F$ containing the section $(a_1,b_1)\mapsto (0,0)$.\footnote{A similar kind of path can be found for any of the local sections of the model.} 

Notice how the possibility of following a `Z' shaped path like the one for the context $\{a_1,b_2\}$ is crucial for the existence of this false positive. This appears to be an aspect common to most of the false positives we know. 
For instance, it is sufficient to invert the labelling of the outcome set $O_{a_1}=\{0,1\}$ for the measurement $a_1$ in the Hardy model to see that a `Z' path is responsible for the false positive in this case as well. This common trait is crucial, and it will essentially motivate our basic strategy to avoid false positives.

Before we illustrate the strategy, we present one last example which will clarify our arguments. Consider the model described by Table \ref{tab: model} and graphically represented in Figure \ref{fig: False}. 
\begin{table}[htbp]
\centering
\begin{tabular}{c c | c c c c}
\hline
$A$ & $B$ & $(0,0)$ & $(1,0)$ & $(0,1)$ & $(1,1)$\\
\hline
$a_1$ & $b_1$ & $1$ & $0$ & $0$ & $1$\\
$a_1$ & $b_2$ & $1$ &  $0$ &  $1$ &  $1$\\
$a_2$ & $b_1$ &  $1$ &  $0$ &  $0$ &  $1$\\
$a_2$ & $b_2$ &  $0$ &  $1$ &  $1$ &  $0$\\
\end{tabular}
\caption{A logically contextual empirical model.}\label{tab: model}
\end{table}

\begin{figure}[htbp]
\centering
\includegraphics[scale=0.6]{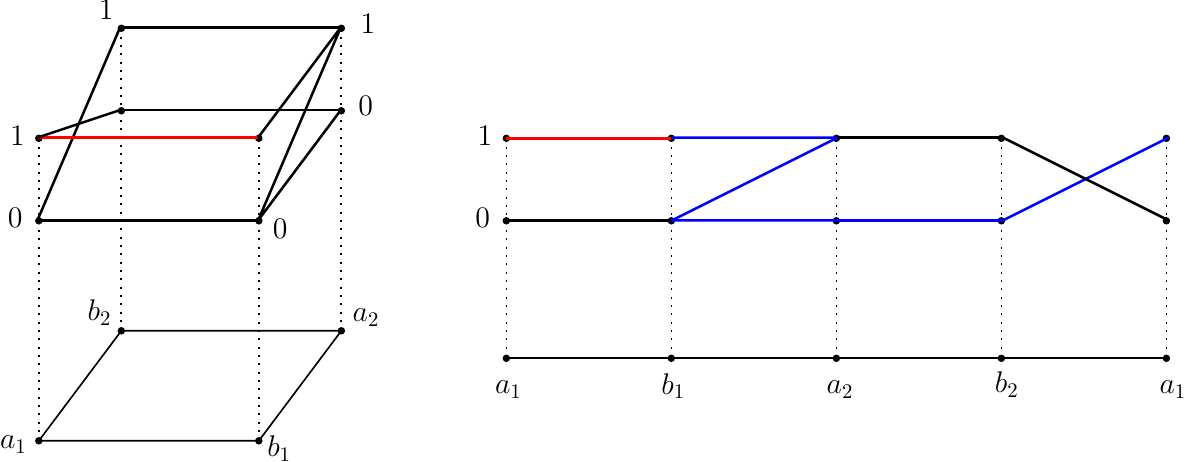}
\caption{A cohomology false positive for section $(a_1,b_1)\mapsto(1,1)$, highlighted in red}\label{fig: False}
\end{figure}
Once again, the section $s:=(a_1,b_1)\mapsto(1,1)$, highlighted in red, is not part of any compatible family, which means that the model is logically contextual at $s$. However, if we allow linear combinations of sections by considering the presehaf $\F:=F_{\mathbb{Z}}\S$, we see that $s$ is part of the compatible family
\[
\begin{split}
\{s,  & [(a_2,b_1)\mapsto (1,1)-(a_2,b_1)\mapsto (1,0)+(a_2,b_1)\mapsto (0,0)],\\
& (a_2, b_2)\mapsto(0,0), (a_1, b_2)\mapsto(0,1) \},
\end{split}
\]
which means that the model is not cohomologically logically contextual at $s$. The reason why this model is important to understand our strategy is the following. We have already mentioned that compatible families correspond to closed loops around the bundle. Suppose that, in the process of trying to extend $s$ to form a closed loop, we could `force' the selection of section $(a_2,b_1)\mapsto (1,1)$ for the context $\{a_2,b_1\}$. This would disallow the `Z' path higlighted in blue in Figure \ref{fig: False}, which is ultimately responsible for the existence of a false positive. It would then be possible to conclude that it is impossible to extend $s$ to a closed loop, even the ones allowed by linear combinations typical of cohomology. Our main strategy will follow exactly this idea. We will derive a series of modified scenarios where each section corresponds to a `forced' selection of sections in adjacent contexts of the original model. 

\section{Changing the original problem}
In order to achieve a complete cohomology invariant for contextuality, we will need to transform the original model into a new one, which is defined on a carefully modified scenario with special properties. The following sections will introduce the main definitions and results.
\subsection{Joint scenarios}

\begin{defn}\label{defn: joint scenario}
Let $\scenario$ be a measurement scenario, such that $\M$ is connected.\footnote{Note that this assumption does not cause any loss of generality. Indeed, contextual behavior in non-connected scenarios can be completely understood by studying the individual connected components of the scenario.} We define the \emph{first joint scenario} of $\scenario$ as the scenario 
\[
\scenario^{(1)}:=\Jscenario{1}.
\]
where
\begin{itemize}
\item $\Meas{1}:=\M$.
\item If $\M$ contains a single context $C$, we let $\Cov{1}:=\{\{C\}\}$.\footnote{This special case will never be used in practice, as we will clarify in Remark \ref{rem: single context}.} Otherwise, we have $|\M|\ge 2$, and define
\[
\Cov{1}:=\left\{\{C,C'\}\subseteq \M\mid C\neq C' \text{ and } C\cap C'\neq \emptyset\right\}
\]
\item For all $C\in \M$, $O_C^{(1)}:=\E(C)$, where $\E:\P(X)^{op}\rightarrow\textbf{Set}$ is the sheaf of events of $\scenario$. 
\end{itemize}
\end{defn}

Note that $\scenario^{(1)}$ is a well-defined measurement scenario, as shown by the following proposition. 
\begin{proposition}\label{prop: scenario}
Let $\scenario$ be a measurement scenario. Then its first joint scenario is well-defined.
\end{proposition}

\begin{proof}
First of all, note that $\Meas{1}$ is finite because $X$ is finite. We only need to show that $\Cov{1}$ is a well-defined measurement cover. We clearly have $\Cov{1}\subseteq\P(\Meas{1})$. 
If $\M$ contains a single context $C$, this is trivially verified. Indeed, $\Cov{1}=\{\{C\}\}$ and we have 
\[
\bigcup_{M\in\Cov{1}}M=\{C\}=\M=\Meas{1}.
\]
Now, suppose $|\M|\ge 2$.  We have $\Cov{1}\subseteq\P(\Meas{1})$, and 
\[
\bigcup_{M\in\Cov{1}}M=\bigcup_{\substack{C,C'\in\M\\ C\cap C'\neq\emptyset}}\{C,C'\}=\M=\Meas{1}.
\]
Indeed, 
\begin{itemize}
\item The inclusion $\bigcup_{M\in\Cov{1}}M\subseteq \M=\Meas{1}$ is trivial, given that each $M\in\Cov{1}$ is included in $\M$ by definition. 
\item Let $C\in\M$. Since $|\M|\ge 2$, there exists a distinct $C'\in \M$. Since $\M$ is connected, there exists a sequence $C=C_0,\dots C_n=C'$ such that $C_i\cap C_{i+1}\neq \emptyset$, hence $C\in\{C_0,C_1\}\subseteq \bigcup_{M\in\Cov{1}}M$.
\end{itemize}
\end{proof}

It is worth spelling out the definition of the sheaf of events of the first joint scenario, which we will denote by $\Event{1}$.
We have $\Event{1}:\P(\Meas{1})^{op}\rightarrow\textbf{Set}$, where, given a $\U\subseteq \Meas{1}$, we have
\[
\Event{1}(\U):=\prod_{C\in \U}O^{(1)}_{C}=\prod_{C\in\U}\E(C),
\]
with restriction maps given by the obvious projections.

To have a better understanding of how the first joint scenario is defined, we give an example in Figure \ref{fig: joint}.
On the left hand side is a simplicial complex representation of the measurement cover 
\[
\M=\{\{a,b,c\},\{b,c,d\},\{a,c,d\},\{a,b,d\},\{b,e,f\},\{e,g\}\}
\]
over the set $X=\{a,b,c,d,e,f,g\}$. On the right hand side, we have the simplicial representation of the cover $\Cov{1}$ of the first joint scenario. 
\begin{figure}[htbp]
\centering
\includegraphics[scale=0.6]{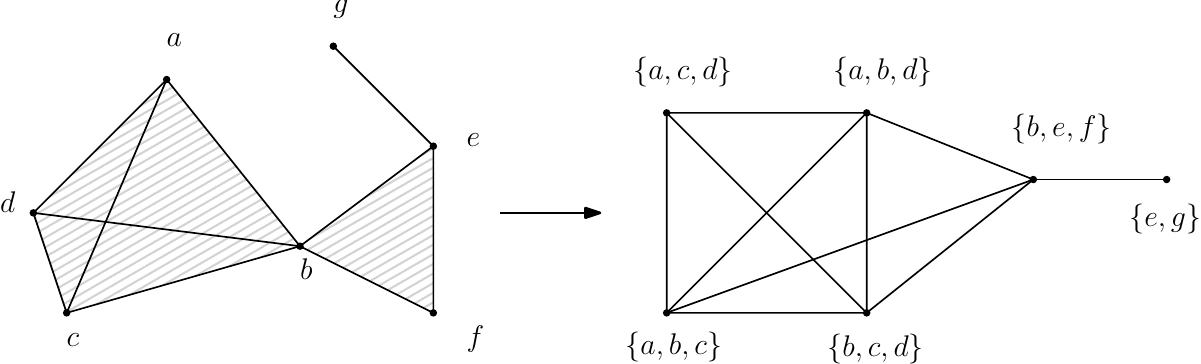}
\caption{A measurement scenario (left) and its first joint version (right). Note that the tetrahedron $\{a,b,c,d\}$ is hollow.}\label{fig: joint}
\end{figure}
Notice that, despite the simplicial complex of the original scenario has dimension $2$, the complex of the first joint scenario is a graph. This is not a coincidence, and it can be easily verified from the definition that every joint scenario has a one-dimensional simplicial complex representation, which can be seen as a graph.

We will often need to repeat the procedure of modifying the original scenario into its joint version. This leads to the following recursive definition:
\begin{defn}\label{defn: k-th joint scenario}
Let $\scenario$ be a measurement scenario, and $k\ge 1$ an integer. We will adopt the following convention:
\[
\scenario^{(0)}:=\scenario.
\]
Then, the \emph{$k$-th joint scenario} of $\scenario$, denoted by $\scenario^{(k)}$, is defined as
\[
\scenario^{(k)}:=\left(\scenario^{(k-1)}\right)^{(1)}.
\]
In other words, it is the first joint scenario of the $(k-1)$--th joint scenario.
\end{defn}
Proposition \ref{prop: scenario} ensures that all the higher-level joint scenarios are well-defined.

\begin{rem}\label{rem: single context}
Consider a scenario $\scenario$. If there exists a $k\ge 0$ such that $\Cov{k}$ contains a single context $\C$, then this necessarily implies that $\scenario$ is \emph{acyclic} in the sense of \cite{Rui}, which means that its associated simplicial complex can be reduced to the empty set by Graham reduction. A result from \cite{Rui}, obtained through an adaptation of Vorob'ev's theorem \cite{Vorobev}, shows that it is impossible to witness contextual behavior in acyclic scenarios. Therefore, from now on, we will always assume $|\Cov{k}|\ge 2$ for all $k\ge 0$.
\end{rem}

\subsection{Joint models}
To an empirical model on a scenario, one can associate an empirical model on the first joint scenario.

\begin{defn}
Let $\S$ be an empirical model on a scenario $\scenario$. The \emph{first joint model} $\Model{1}$ is a possibilistic empirical model on the scenario $\scenario^{(1)}$, defined as follows: for all $\U\subseteq \Meas{1}$, we have
\[
\Model{1}(\U):=\setdef{(s_C)_{C\in \U}\in\prod_{C\in \U}\S(C)}{s_C\mid_{C\cap C'}=s_{C'}\mid_{C\cap C'}~\forall C,C'\in\U}.
\]
The restriction maps are inherited from $\Event{1}$.
\end{defn}
Note that, in particular, for elements $\C=\{C,C'\}\subseteq\M$ of the cover $\Cov{1}$, $\Model{1}(\C)$ coincides with the following pullback:

\begin{center}
\begin{tikzpicture}[auto]
\matrix (m) [matrix of math nodes, row sep=2.5em, column sep=2.5em, text height=1.5ex, text depth=0.25ex]
{
\Model{1}(\C) & & \S(C)\\
& \scalebox{1.5}{$\lrcorner$} &\\
\S(C') & & \S(C\cap C')\\
};
\path[-stealth, font=\scriptsize]
(m-1-1) edge node[auto] {} (m-1-3)
(m-1-1) edge node[auto] {} (m-3-1)
(m-1-3) edge node[auto] {$\rho_{C\cap C'}^C$} (m-3-3)
(m-3-1) edge node[auto] {$\rho_{C\cap C'}^{C'}$} (m-3-3);
\end{tikzpicture}
\end{center}


The following proposition shows that the first joint model is a well-defined concept.

\begin{proposition}\label{prop: model}
Let $\S$ be an empirical model. Then, $\Model{1}$ is a well-defined empirical model.
\end{proposition}

\begin{proof}
First of all, note that $\Model{1}$ is a subpresheaf of $\Event{1}$. Indeed, 
\[
\Model{1}(\U)\subseteq\prod_{C\in\U}\S(C)\subseteq\prod_{C\in\U}\E(C)=\Event{1}(\U).
\]
Now, we need to verify conditions \ref{cond: 1}, \ref{cond: flasque} and \ref{cond: 3} of the definition of an empirical model given in section \ref{sec: basics}.

\begin{enumerate}
\item Let $\C=\{C,C'\}\in\Cov{1}$. Because $\S$ is an empirical model, we now that $\S(C)\neq\emptyset$, given that $C\in\M$. Let $s_C\in\S(C)$. Since $\S$ is flasque beneath the cover, and because $C\cap C'\subseteq C'\in\M$, the restriction map $\rho_{C\cap C'}^{C'}:\S(C')\rightarrow\S(C\cap C')$ is surjective. Therefore, there exists $s_{C'}\in\S(C')$ such that 
\[
\rho_{C\cap C'}^{C'}(s_{C'})=s_{C'}\mid_{C\cap C'}=s_C\mid_{C\cap C'}.
\]
Hence, $(s_C,s_{C'})\in\Model{1}(\{C, C'\})$. 
\item Let $\U\subseteq \U'\subseteq \C$ for some context $\C=\{C,C'\}\in\Cov{1}$. There are three cases,\footnote{If $\U=\emptyset$ or $\U'=\emptyset$, then the condition is trivially verified, as $\rho_\emptyset^\U:: s\mapsto \ast$ for all $s\in\Model{1}(\U)$.} of whom only one is non-trivial:

\begin{itemize}
\item $\U=\U'=\{C\}$ (or $\{C'\}$). In this case, $\rho^\U_{\U'}$ is the identity, which is obviously surjective.
\item $\U=\U'=\C$. Once again, $\rho^\U_{\U'}$ is the identity.
\item Suppose $\U=\{C\}$, and $\U'=\C$. Let $s_C\in\S(C)$. Becuase $\S$ is flasque beneath the cover, the restriction map $\rho_{C\cap C'}^{C'}:\S(C')\rightarrow\S(C\cap C')$ is surjective. Hence, there exists a $s_{C'}\in\S(C')$ such that $s_{C'}\mid_{C\cap C'}=s_C\mid_{C\cap C'}$. Thus, $(s_C,s_{C'})\in\Model{1}(\C)$, and
\[
(s_C,s_{C'})\mid_\U=(s_C,s_{C'})\mid_{\{C\}}=s_C,
\]
which shows that $\rho_\U^{\U'}$ is surjective.
\end{itemize}

\item Let $F:=\left\{(s_C,s_{C'})_\C\right\}_{\C\in\Cov{1}}$ be a compatible family\footnote{A more precise notation would be $\left\{(s_C,s_{C'})_{\{C, C'\}}\right\}_{\{C, C'\}\in\Cov{1}}$, but we will often use the one we adopt here to simplify notation.} for $\Model{1}$, which means that $(s_C,s_{C'})_\C\in\Model{1}(\C)$ for all $\C\in\Cov{1}$, and
\[
(s_C,s_{C'})\mid_{\C\cap \D}=(s_D,s_{D'})\mid_{\C\cap \D}
\]
for all $\C=\{C,C'\}$ and $\D=\{D,D'\}$ in $\Cov{1}$.\footnote{Explicitly, excluding the trivial case where $\C\cap \D=\emptyset$, compatibility implies that, if $\C\cap \D=\{C\}$, e.g. in the case where $C=D$, we must have $s_C=s_D$ (or similarly for the other cases). In other words, $F$ cannot contain two different local sections of $\S$ at the same context.\label{footnote: 1}}
The family $F$ induces the global section 
\[
g:=(s_C)_{C\in\M}\in\prod_{C\in\M}\S(C)\subseteq\Model{1}\left(\Meas{1}\right),
\]
which is well-defined by the remark in footnote \ref{footnote: 1}. 
The fact that $g_C\mid_{C\cap C'}=g_{C'}\mid_{C\cap C'}$ for all $C,C'\in \M$ is trivially verified given that $g_C=s_C$ and $(s_C,s_{C'})\in\Model{1}(\{C,C'\})$. 
\end{enumerate}
\end{proof}

We will often need to repeat the procedure of taking the first joint model, which leads to the following definition (cf. Definition \ref{defn: k-th joint scenario}).

\begin{defn}
Let $\S$ be an empirical model on a scenario $\scenario$ and let $k\ge 1$ be an integer. We will adopt the convention $\Model{0}:=\S$. 
Then, the \emph{$k$-th joint model of $\S$}, denoted by $\Model{k}$, is defined by
\[
\Model{k}:=\left(\Model{k-1}\right)^{(1)}.
\]
In other words, it is the first joint model of the $(k-1)$-th joint model of $\S$.
\end{defn}

Proposition \ref{prop: model} guarantees that all the higher-level joint models are well-defined concepts. 

We end this section with two important remarks.

\begin{rem}\label{rem: subtlety}
Consider an emirical model $\S$ on a scenario $\scenario$. There is a key subtlety in the definition of the joint models of $\S$ which we will exploit in some of the proofs in Section \ref{sec: main}. Let $C\in\Cov{k-1}$.
The subtlety consists in the following equality, which simply follows by definition:
\begin{equation}\label{equ: subtlety}
\Model{k}(\{C\})=\Model{k-1}(C).
\end{equation}
Consider two contexts $\C_1=\{C_1^1,C_1^2\},\C_2=\{C_2^1,C_2^2\}\in\Cov{k}$, such that $\C_1\cap \C_2\neq\emptyset$. W.l.o.g. we can suppose that $\C_1\cap\C_2=\{C_1^2\}$ (i.e. $C_1^2=C_2^1$). Suppose we have a section $s_{\C_1}\in\Model{k}(\C_1)$. Then, because of \eqref{equ: subtlety}, the restricted section 
\[
s_{\C_1}\mid_{\C_1\cap \C_2}=s_{\C_1}\mid_{\{C_1^2\}}
\]
can be seen both as an element of $\Model{k}(\{C_1^2\})$, or, equivalently, as an element of $\Model{k-1}(C_1^2)$. In the latter case, we will denote the restricted section as
\[
s_{\C_1}\mid_{C_1^2}\in \Model{k-1}(C_1^2).
\]

\begin{rem}\label{rem: flatten}
Let $\S$ be an empirical model on a scenario $\scenario$. By definition, the possible sections of $\Model{1}$ are pairs of sections of $\S$. Similarly, sections of $\Model{2}$ are \emph{pairs of pairs} of sections of $\S$. In general, sections of $\Model{k}$ are \emph{pairs of pairs ... of pairs} ($k$ times) of sections of $\S$. For our purposes, given a section $s$ of $\S$, we will need to list those sections of $\Model{k}$ that \emph{contain} $s$. To do this, we will use the $\flatten$ function, whose name is borrowed from popular programming languages. This function takes a section $t_\C\in\Model{k}(\C)$ (which is a pair of pairs ... of pairs (k times) of sections of $\Model{k}$) as argument and returns a single set containing all the sections of $\Model{k}$ that appear in $t_\C$.
For instance, for $k=3$, we have
\[
\flatten\left[(((s_1,s_2),(s_3,s_4)),((s_5,s_6),(s_7,s_8)))\right]=\{s_1,s_2,s_3,s_4,s_5,s_6,s_7,s_8\}
\]

\end{rem}

\end{rem}

\subsection{Interpretation and examples}\label{sec: interpretation}
At the end of the previous section, we briefly sketched our strategy to avoid false positives, which consists of considering multiple compatible local sections at the same time, instead of focusing on a single one. The notion of joint model perfectly embodies this idea. As discussed above, local sections of the first joint model of an empirical model $\S$ are pairs of sections of $\S$ above adjacent contexts. This allows one to `force' the selection of the sections on adjacent contexts in the original model, thus reducing the chances of the existence of a false positive. Higher-level joint models further refine this approach and allow to consider three, four, $k$ compatible sections at the same time. 
These statements will be made precise in Section \ref{sec: main}, but it is worth giving some examples that will guide us through the technical results. 

Let us start by illustrating the first joint model of the Hardy model \ref{tab: Hardy}. Recall that $X=\{a_1,a_2,b_1,b_2\}$, $\M=\{\{a_1,b_1\},\{a_1,b_2\},\{a_2,b_1\},\{a_2,b_2\}\}$ and $O_m=\{0,1\}$ for all $m\in X$. Let $C_1:=\{a_1,b_1\}$, $C_2:=\{a_1,b_2\}$, $C_3:=\{a_2,b_1\}$ and $C_4:=\{a_2,b_2\}$. Then we have $\Meas{1}=\M$ and
\[
\Cov{1}=\{\{C_1,C_2\},\{C_2,C_4\},\{C_3,C_4\},\{C_1,C_3\}\}.
\]
In Table \ref{tab: enum1}, we introduce an enumeration of the possible sections for the Hardy model (blank spaces correspond to impossible sections).
\begin{table}[htbp]
\centering
\begin{tabular}{c c | c c c c}
\hline
$A$ & $B$ & $(0,0)$ & $(1,0)$ & $(0,1)$ & $(1,1)$\\
\hline
$a_1$ & $b_1$ & $s_1$ & $s_2$ & $s_3$ & $s_4$\\
$a_1$ & $b_2$ &  &  $s_5$ &  $s_6$ &  $s_7$\\
$a_2$ & $b_1$ &   &  $s_8$ &  $s_9$ &  $s_{10}$\\
$a_2$ & $b_2$ &  $s_{11}$ &  $s_{12}$ &  $s_{13}$ & \\
\end{tabular}
\caption{An enumeration of the possible sections of the Hardy model.}\label{tab: enum1}
\end{table}
Thanks to this notation, we can represent the planar bundle diagram for the first joint model of the Hardy model in Figure \ref{fig: FJM Hardy}

\begin{figure}[htbp]
\centering
\includegraphics[scale=0.5]{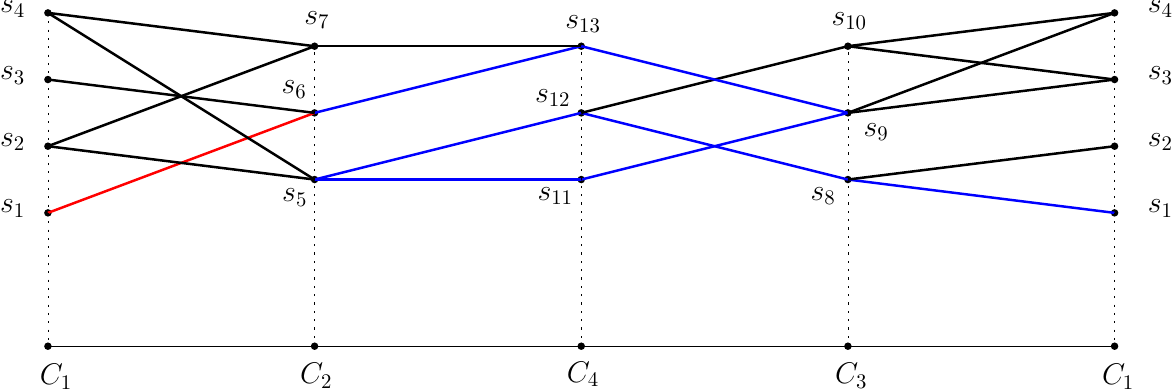}
\caption{The first joint model of the Hardy model. In red, the only section containing $s_1$ in the context $\{C_1, C_2\}$. In blue, a cohomology loop containing $s_1$.}\label{fig: FJM Hardy}
\end{figure}
Compare this to Figure \ref{fig: Hardy}, where we highlighted in red the section $s_1$, which is not part of any compatible family. The only section in the first joint model containing $s_1$ is $(s_1,s_6)$, marked in red in Figure \ref{fig: FJM Hardy}. Notice that this section is not part of any compatible family in the joint model either. In Figure \ref{fig:  Hardy counter}, we provided a cohomology loop containing $s_1$, which is responsible for the existence of a false positive. Note that, in the case of the first joint model, it is no longer possible to create `Z' shaped paths above a single context (this fact is not a coincidence, as we will se in Lemma \ref{lem: no Z} and more generally in Theorem \ref{thm: main}), however, it is still possible to find a more complex cohomology loop containing $(s_1,s_6)$, namely
\[
\begin{split}
\{
(s_1,s_6),(s_6,s_{13})-(s_5,s_{11})+(s_5,s_{12}), (s_{13},s_9)-(s_{11},s_9)+(s_{12},s_8),(s_8,s_1)
\},
\end{split}
\]
which is highlighted in blue in Figure \ref{fig: FJM Hardy}. 

Let us now consider the model of Table \ref{tab: model}. In Table \ref{tab: enum2}, we give an enumeration of its possible sections.

\begin{table}[htbp]
\centering
\begin{tabular}{c c | c c c c}
\hline
$A$ & $B$ & $(0,0)$ & $(1,0)$ & $(0,1)$ & $(1,1)$\\
\hline
$a_1$ & $b_1$ & $s_1$ &  & & $s_2$\\
$a_1$ & $b_2$ & $s_3$ &  & $s_4$  &  $s_5$\\
$a_2$ & $b_1$ &  $s_6$ &  &   &  $s_7$\\
$a_2$ & $b_2$ &   &  $s_8$ &  $s_9$ & \\
\end{tabular}
\caption{An enumeration of the possible sections of the model \ref{tab: model}.}\label{tab: enum2}
\end{table}
With this enumeration, we illustrate the first joint model as a planar bundle diagram in Figure \ref{fig: FJM model}.

\begin{figure}[htbp]
\centering
\includegraphics[scale=0.5]{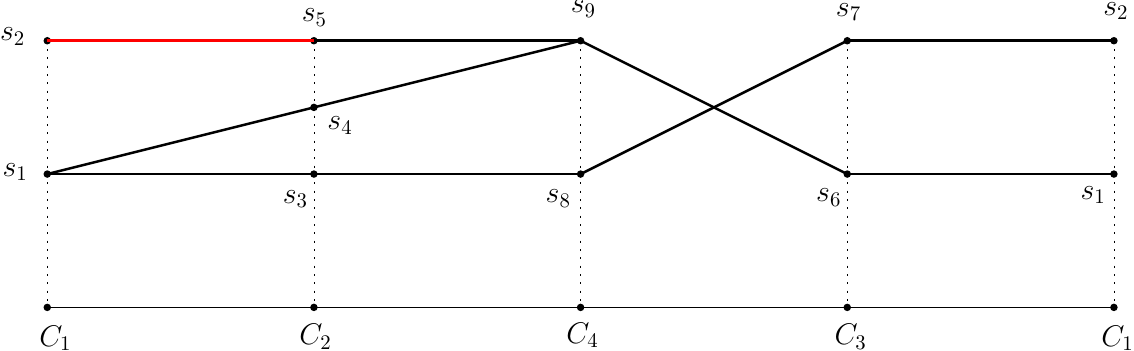}
\caption{The first joint model of the model given by Table \ref{tab: model}. In red, the section S}\label{fig: FJM model}
\end{figure}
We have already shown that $s_2$ is not part of any compatible family, but it is part of a cohomology loop, which gives rise to a false positive. In the joint model, the only section containing $s_2$ is $(s_2, s_5)$. Note that not only $(s_2, s_5)$ is not part of any compatible family, but it appears not to be part of any cohomology loop either. We have successfully removed the `Z' path responsible for the false positive. This fact perfectly reflects the discussion on this model carried out at the end of section \ref{sec: false positives}. By imposing the joint selection of $s_1$ \emph{and} $s_5$, we have successfully removed the false positive. A formal proof of this fact will be given in Section \ref{sec: main}. 

It is fairly easy to see that the fact that the first joint model is logically contextual at $(s_2, s_5)$, which is the only section containing $s_2$, implies that the underlying model is logically contextual at $s_2$. However, the relation between the contextual properties of joint models and the original ones may not be immediately clear. We will give all the details about this question in the following section.

\section{The contextuality of joint models}\label{sec: contextuality}
By looking at the examples of the previous sections, a natural question to ask is: what conclusions can we draw on an empirical model by looking at the contextual properties of its joint models? The answer is given by the following results.

\begin{proposition}\label{prop: contextuality of FJM}
Consider an empirical model $\S$ on a measurement scenario $\scenario$. Let $C\in \M$ and $s\in\S(C)$. The following are equivalent.
\begin{enumerate}
\item The model $\S$ is logically contextual at $s$. \label{stat1}
\item There exists a $C'\in\M$ with $C\cap C'\neq\emptyset$ such that, for all $t\in\S(C')$ verifying $t\mid_{C\cap C'}=s\mid_{C\cap C'}$, we have $\LCs[\Model{1}]{(s,t)}$ \label{stat2}
\item For all $C'\in\M$ with $C\cap C'\neq\emptyset$, for all $t\in\S(C')$ verifying $t\mid_{C\cap C'}=s\mid_{C\cap C'}$, we have $\LCs[\Model{1}]{(s,t)}$. \label{stat3}
\end{enumerate}
\end{proposition}
\begin{proof}
The fact that \ref{stat3} implies \ref{stat2} is trivial. 
\begin{itemize}

\item \ref{stat2} $\Rightarrow$ \ref{stat1}: We will prove $\neg$\ref{stat1} $\Rightarrow\neg$\ref{stat2}. Suppose $\neg\LCs{s}$. Then there exists a family $F:=\{s_C\in\S(C)\}_{C\in\M}$, compatible for $\S$, such that $s_C=s$. We want to show that, for all $C'\in\M$ with $C\cap C'\neq\emptyset$, there exists $t\in\S(C')$ verifying $t\mid_{C\cap C'}=s\mid_{C\cap C'}$ such that $\neg\LCs[\Model{1}]{(s,t_{C'})}$. 

Consider the family 
\[
F':=\left\{(s_K,s_{K'})\in\Model{1}(\{K,K'\})\right\}_{\{K,K'\}\in\Cov{1}}.
\] 
This family is well-defined (i.e. $(s_K,s_{K'})$ is indeed in $\Model{1}(\{K,K'\})$) by compatibility of $F$, and it is compatible for $\Model{1}$ by definition. Let $C'\in\M$ with $C\cap C'\neq\emptyset$, and consider $t:=s_{C'}\in\S(C')$. Then $(s,t)=(s_C,s_{C'})\in F'$, which proves that $t\mid_{C\cap C'}=s\mid_{C\cap C'}$ (as $t=s_{C'}$ and $s=s_C$), and $\neg\LCs[\Model{1}]{(s,t_{C'})}$.

\item \ref{stat1} $\Rightarrow$ \ref{stat3}: We will prove $\neg$\ref{stat3} $\Rightarrow\neg$\ref{stat1}. Suppose there exists a $C'$ with $C\cap C'\neq\emptyset$, such that there exists a $t\in\S(C')$, with $t\mid_{C\cap C'}=s\mid_{C\cap C'}$, verifying $\neg\LCs[\Model{1}]{(s,t)}$. This means that there exists a family 
\[
F:=\left\{(v_K,v_{K'})\in\Model{1}(\{K,K'\})\right\}_{\{K,K'\}\in\Cov{1}},
\]
compatible for $\Model{1}$, such that $(v_C,v_{C'})=(s,t)$. Consider the family $F':=\{v_K\in\S(K)\}_{K\in\M}$. This family contains precisely one local section for each context of $\M$ by connectedness of the cover. Moreover, each such global section is well-defined by compatibility of $F$ (see footnote \ref{footnote: 1}). F' is a compatible family for $\S$. Indeed, given $K,K'\in\M$, because $(v_K,v_{K'})\in \Model{1}(\{K,K'\})$, we must have $v_K\mid_{K\cap K'}=v_{K'}\mid_{K\cap K'}$. Therefore, because $v_C=s$, the section $s$ is contained in the compatible family $F'$, proving that $\neg\LCs{s}$. 

\end{itemize}

\end{proof}

\begin{corollary}\label{cor: SC}
Let $\S$ be an empirical model on a scenario $\scenario$. Then $\S$ is strongly contextual if and only if $\Model{1}$ is strongly contextual.
\end{corollary}
\begin{proof}
Suppose $\SCs$. Consider two contexts $C,C'\in\M$ such that $C\cap C'\neq\emptyset$ (these always exist by connectedness of $\M$ and the fact that $|\M|\ge 2$). Take an arbitrary section $(s_C,t_{C'})\in\Model{1}(\{C,C'\})$ (there is at least one such a section by condition \ref{cond: 1} of the definition of an empirical model). We want to show that $\Model{1}$ is logically contextual at $(s_C,t_{C'})$. Since $\S$ is strongly contextual, it is in particular logically contextual at $s_C$. By Proposition \ref{prop: contextuality of FJM}, this implies that, for all $C'\in\M$ with $C\cap C'\neq\emptyset$, for all $t\in\S(C')$ verifying $t\mid_{C\cap C'}=s_C\mid_{C\cap C'}$, we have $\LCs[\Model{1}]{(s_C,t)}$. In particular, if we take $t:=t_{C'}$, we have $\LCs[\Model{1}]{(s_C,t_{C'})}$. 

For the converse, suppose $\SCs[\Model{1}]$. Let $C\in\M$ and take an arbitrary section $s\in\S(C)$. Let $C'\in\M$ such that $C\cap C'\neq\emptyset$ (these always exist by the usual assumptions). Because $\SC(\Model{1})$, we know that for all $t\in\S(C')$ verifying $t\mid_{C\cap C'}=s\mid_{C\cap C'}$, we have $\LCs[\Model{1}]{(s,t)}$. By Proposition \ref{prop: contextuality of FJM}, we conclude that $\LCs{s}$. 
\end{proof}

Proposition \ref{prop: contextuality of FJM} motivates the following definition.

\begin{defn}\label{defn: LCk}
Let $\S$ be an empirical model on a scenario $\scenario$. Let $s\in\S(C)$ be a local section at some context $C\in\M$, and $k\ge 1$. We say that \emph{$\S$ is $\LCk{k}$ at $s$}, and write $\LCk{k}(\S,s)$, if we have $\LC\left(\Model{k}, t\right)$ for all local section $t$ of $\Model{k}$ such that $s\in\flatten(t)$ (cf. Remark \ref{rem: flatten}).
\end{defn}

By applying simple inductive arguments, we immediately have the following additional corollaries of Proposition \ref{prop: contextuality of FJM}:

\begin{corollary}\label{cor: LCk}
Let $\S$ be an empirical model on a scenario $\scenario$. Let $C\in \M$ and $s\in\S(C)$. Then the following are equivalent:
\begin{enumerate}
\item $\S$ is logically contextual at $s$.
\item There exists a $k\ge 1$ such that $\LCk{k}(\S,s)$
\item $\LCk{k}(\S,s)$ for all $k\ge 1$. 
\end{enumerate}
\end{corollary}

\begin{corollary}\label{cor: SC2}
Let $\S$ be an empirical model on a scenario $\scenario$. The following are equivalent:
\begin{enumerate}
\item $\S$ is strongly contextual.
\item There exists a $k\ge1$ such that $\Model{k}$ is strongly contextual
\item $\Model{k}$ is strongly contextual for all $k\ge 0$
\end{enumerate}
\end{corollary}

%

Note that there is no need to extend Definition \ref{defn: LCk} to strong contextuality as this would be equivalent to regular strong contextuality by Corollary \ref{cor: SC2}.

\section{Cyclic models and their properties}\label{sec: main}

Before we prove the main results of the paper, we will need to introduce the notions of \emph{path}, \emph{cycle}, and \emph{cyclic model}, and thoroughly inspect their properties. We start with an important remark:

\begin{rem}
Let $\scenario$ be a measurement scenario. By definition, for each $k\ge 1$, the contexts of $\Cov{k}$ are sets of contexts of $\Cov{k-1}$. In order to avoid confusion between the contexts of $\Cov{k}$ and those of $\Cov{k-1}$ we will denote them using different calligraphic styles. The typical hierarchy we will use is the following: 
\[
c\in\Cov{k-2}\rightarrow C\in\Cov{k-1}\rightarrow \C\in\Cov{k}\rightarrow \mathfrak{C}\in\Cov{k+1}.
\]
Note that the hierarchy will always be the same, but we will \emph{not} fix a calligraphic style for a specific $k$, as we will have to deal with many different cases.

\end{rem}

\subsection{Paths and cycles}

Let us inspect some of the properties of joint scenarios. We start by introducing the notion of \emph{path}. 

\begin{defn}\label{defn: sequence}
Let $\scenario$ be a measurement scenario and $n,k\ge 1$. An \emph{$n$-path for $\Cov{k}$} is a set $\D:=\{C_1,\dots, C_n\}\subseteq\Cov{k-1}$ of $n$ distinct contexts of $\Cov{k-1}$ such that $C_i\cap C_{i+1}\neq\emptyset$ for all $1\leq i\leq n-1$. It is called a \emph{cycle} if, in addition, $C_n\cap C_1\neq \emptyset$. An $n$-path $\D$ is called \emph{chordal} if there exist two non-consecutive indices $i,j$, with $\{i,j\}\neq\{1,n\}$, such that $C_i\cap C_j\neq\emptyset$.  
\end{defn}

We can think of an $n$-path for $\Cov{k}$ as a sequence of distinct vertices in the graph generated by $\Cov{k}$. This corresponds to the graph-theoretic notion of \emph{simple path}. 
Similarly, (chordal) cycles for $\Cov{k}$ correspond to \emph{(chordal) simple cycles} in graph theory. 
In Figure \ref{fig: paths examples}, we give some graphical examples.

\begin{figure}[htbp]
\centering
\includegraphics[scale=0.6]{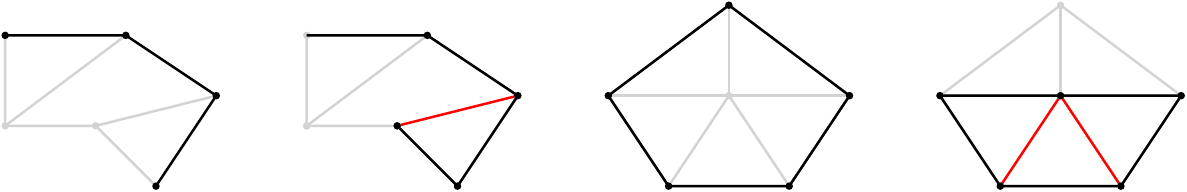}
\caption{Different types of paths for $\Cov{k}$. The grey graph represents $\Cov{k}$. From left to right: a chordless $4$-path, a chordal $4$-path (chord highlighted in red), a chordless $5$-cycle and a chordal $5$-cycle (chords highlighted in red).}\label{fig: paths examples}
\end{figure}

\begin{rem}
In graph theory, a simple path can be equivalently described by the sequence of edges connecting the vertices. Similarly, an $n$-path $\D_\bullet=\{C_1,\dots, C_n\}\subseteq\Cov{k-1}$ for $\Cov{k}$ can be specified by the set
\[
\underline{\mathfrak{D}}=\{\{C_1,C_2\},\{C_2,C_3\},\dots,\{C_{n-1},C_n\}\}\subseteq\Cov{k},
\]
containing contexts of $\Cov{k}$, i.e. edges of the graph generated by $\Cov{k}$. The set $\underline{\mathfrak{D}}$ will be referred to as the \emph{edge representation} of the path $\D_\bullet$. To avoid confusion, from now on, we will denote $\D_\bullet$ for the vertex representation and $\underline{\mathfrak{D}}$ for the edge representation.
\end{rem}

\subsubsection{$3$-cycles: proper and improper}

Cycles for $\Cov{k}$ of size $3$ present some peculiarities that deserve to be discussed in details in order to avoid confusion. The reason is that, although they are technically chordless, one of their edges could be seen as a chord connecting the remaining two. A key aspect of chordless $n$-cycles for $\Cov{k}$ of size $n\ge 4$, which will be proved in Proposition \ref{prop: cycle2}, is that they must be generated by $n$-cycles for $\Cov{k-1}$. This is not generally true for $3$-cycles. Indeed, we could potentially have a $3$-cycle $\D_\bullet=\{C_1,C_2,C_3\}$ for $\Cov{k}$ which is generated by a star-shaped configuration of the $C_i$'s, seen as edges of $\Cov{k-1}$, as shown in Figure \ref{fig: triangle}.

\begin{figure}[htbp]
\centering
\includegraphics[scale=0.6]{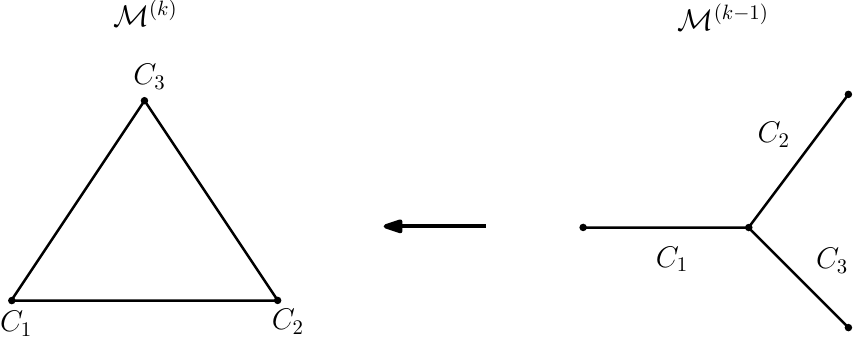}
\caption{A non proper $3$-cycle $\D_\bullet=\{C_1,C_2,C_3\}$ for $\Cov{k}$.}\label{fig: triangle}
\end{figure}

We will refer to this kind of $3$-cycles as \emph{improper $3$ cycles for $\Cov{k}$}. On the other hand, a \emph{proper $3$-cycle} $\D_\bullet=\{C_1,C_2,C_3\}\subseteq\Cov{k-1}$ for $\Cov{k}$ is a $3$-cycle for $\Cov{k}$ such that $\underline{D}=\{C_1,C_2,C_3\}$ is a $3$-cycle for $\Cov{k-1}$, as shown in Figure \ref{fig: proper triangle}.

\begin{figure}[htbp]
\centering
\includegraphics[scale=0.6]{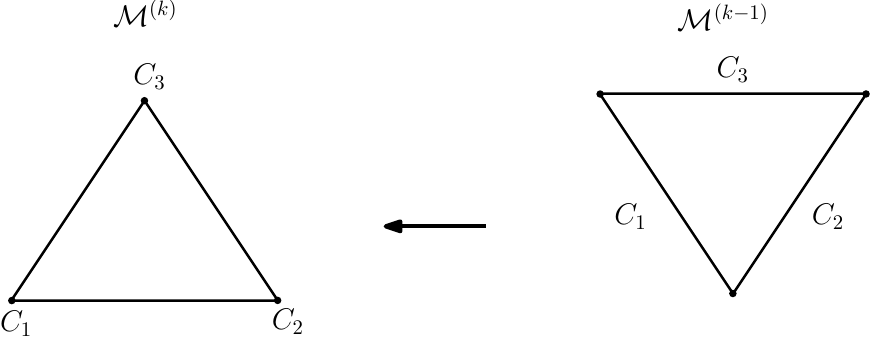}
\caption{A proper $3$-cycle $\D_\bullet=\{C_1,C_2,C_3\}$ for $\Cov{k}$.}\label{fig: proper triangle}
\end{figure}

\subsubsection{Fundamental properties of paths and cycles}

In this paragraph, we will present some key properties of paths and cycles that will play a crucial role in the proofs of the main results of the paper.
The first proposition, for instance, shows that paths and cycles are preserved when taking the joint version of a scenario, in the sense that they naturally give rise to paths and cycles in the new scenario.

\begin{proposition}\label{prop: cycle1}
Let $\scenario$ be a measurement scenario. Let $k\ge1$, and let $\underline{\D}:=\{C_1,\dots, C_n\}\subseteq\Cov{k}$ be an $n$-path for $\Cov{k}$. Then, the set 
\[
\underline{\mathfrak{D}'}=\{\K_1,\dots, \K_{n-1}\}:=\{\{C_1,C_2\},\{C_2,C_3\},\dots,\{C_{n-1},C_n\}\}
\]
is an $(n-1)$-path for $\Cov{k+1}$. 
Moreover, if $\underline{\D}$ is a cycle, then 
\[
\underline{\mathfrak{D}'} =\{\K_1,\dots, \K_n\}:=\{\{C_1,C_2\},\{C_2,C_3\},\dots,\{C_{n-1},C_n\},\{C_n,C_1\}\} 
\]
is a chordless $n$-cycle for $\Cov{k+1}$.
\end{proposition}

\begin{proof}
The elements of $\underline{\mathfrak{D}}'$ are all distinct because the elements of $\D$ are all distinct. Moreover, $\underline{\mathfrak{D}}'\subseteq\Cov{k}$ because $C_i\cap C_{i+1}\neq\emptyset$ for all $1\leq i\leq n-1$. 
We have $\K_i\cap \K_{i+1}=\{C_i,C_{i+1}\}\cap\{C_{i+1},C_{i+2}\}=\{C_{i+1}\}\neq\emptyset$ for all $1\leq i\leq n-2$. If $\underline{\D}$ is a cycle, then $C_n\cap C_1\neq\emptyset$, thus $\{C_n,C_1\}\in\Cov{k}$, and we have $\K_n\cap \K_1=\{C_n,C_1\}\cap\{C_1,C_2\}=\{C_1\}\neq\emptyset$, which proves that $\underline{\mathfrak{D}}'$ is a cycle. To prove that it is chordless, suppose by contradiction that there exist two non-consecutive indices $i,j$, with $\{i,j\}\neq\{1,n\}$, such that $\K_i\cap \K_j\neq\emptyset$. Then $\{C_i,C_{i+1}\}\cap\{C_j,C_{j+1}\}\neq\emptyset$, which contradicts the fact that the $C_i$'s are all distinct. 
\end{proof}

Consider a measurement scenario $\scenario$ and let $k\ge 2$, $n\ge 3$. Let $\D_\bullet:=\{C_1,\dots, C_n\}\subseteq\Cov{k-1}$ be an $n$-cycle for $\Cov{k}$. By definition of a cycle, we know that there exist $k_1,\dots k_n\in\Cov{k-2}$ such that $\{k_n\}=C_n\cap C_1$ and $\{k_i\}=C_i\cap C_{i+1}$ for all $1\leq i\leq n-1$. With this notation, we can prove the following proposition.

\begin{proposition}\label{prop: cycle2}
Let $\scenario$ be a measurement scenario, let $k\ge 2$ and $n\ge 3$. Let $\D_\bullet:=\{C_1,\dots, C_n\}\subseteq\Cov{k-1}$ be chordless $n$-cycle for $\Cov{k}$ such that it is not an improper $3$-cycle. Then the set 
\[
D_\bullet'=\{k_1,\dots k_n\}\subseteq\Cov{k-2}
\]
is an $n$-cycle for $\Cov{k-1}$.
\end{proposition}
\begin{proof}
First of all, we need to verify that the $k_i$'s are all disinct. Suppose there are two indices $1\leq i,j\leq n$ such that $k_i=k_j$. Then $i$ and $j$ must be consecutive because otherwise we would have $C_i\cap C_j=\{k_i\}\neq\emptyset$, which contradicts the fact that $\D_\bullet$ is chordless. Thus, we only need to prove that $k_i\neq k_n$ and that $k_i\neq k_{i+1}$ for all $1\leq i\leq n-1$. Suppose $1\leq i\leq n-2$ and assume $k_i=k_{i+1}$, then we have 
\[
C_i\cap C_{i+1}=\{k_i\}=\{k_{i+1}\}=C_{i+1}\cap C_{i+2},
\]
which implies that $C_i\cap C_{i+2}=\{k_i\}\neq\emptyset$ which contradicts the fact that $\D_\bullet$ is chordless and not an improper $3$-cycle. We can prove that $k_{n-1}\neq k_n$ and $k_n\neq k_1$ with the same argument. 

We are only left to prove that consecutive $k_i$'s intersect. Because the $k_i$'s are all distinct, we know that $C_i=\{k_{i-1},k_i\}$ for all $2\leq i\leq n$, and $C_1=\{k_n,k_1\}$. Since $C_1,\dots, C_n\in\Cov{k-1}$, this implies that $k_{i-1}\cap k_i\neq\emptyset$ for all $2\leq i\leq n$, and $k_1\cap k_n\neq\emptyset$. 
\end{proof}

We define the notion of \emph{cyclic scenario}. 

\begin{defn}\label{defn: cyclic scenario}
A measurement scenario $\scenario$ is called \emph{cyclic} if $\Cov{1}$ is a chordless cycle (in edge representation).
\end{defn}

Thanks to Proposition \ref{prop: cycle1}, we immediately have the following: 
\begin{proposition}\label{prop: cyclicn}
Let $\scenario$ be a cyclic scenario, and let $n:=|\M|$. Then $\Cov{k}$ is a chordless $n$-cycle for all $k\ge 1$. In particular, $\scenario^{(l)}$ is cyclic for all $l\ge 0$. 
\end{proposition}

\begin{rem}\label{rem: notation}
Before we proceed, we shall introduce a convention concerning notation. Suppose we have a chordless $n$-path $\D_\bullet:=\{C_1,\dots, C_n\}$ for $\Cov{k}$, with $n\ge 1$ and $k\ge 2$. Suppose, in addition, that $\D_\bullet$ is not an improper $3$-cycle. Proposition \ref{prop: cycle2} shows that we can relabel the components $c_i^1,c_i^2$ of each $C_i=\{c_i^1, c_i^2\}$ in such a way that $c_i^2=c_{i+1}^1$ for all $1\leq i\leq n-1$. If $\D$ is a cycle, we also have $c_n^2=c_1^1$.\footnote{In other words, we relabel the components of the $C_i$'s in such a way that $k_i=c_i^2$ for all $1\leq i\leq n$, where $k_1,\dots k_n$ are defined as in Proposition \ref{prop: cycle2}}
This notation will be used extensively in many of the proofs of this section. To clarify how it is constructed, we provide a graphical representation in Figure \ref{fig: path}.

\begin{figure}[htbp]
\centering
\includegraphics[scale=0.6]{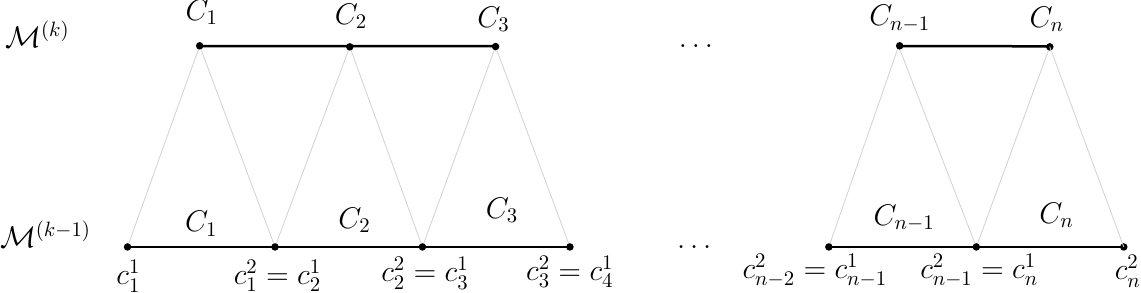}
\caption{The standard notation for $n$-paths.}\label{fig: path}
\end{figure}

%
%
%
%
%

\end{rem}

Suppose we have a cyclic scenario, and a $n$-path $\underline{\mathfrak{D}}:=\{\C_1,\dots \C_n\}\subseteq\Cov{k}$ for $\Cov{k}$, with $n<|\M|$. Because $\scenario$ is chordless, we know by Proposition \ref{prop: cyclicn} that $\Cov{k}$ is a chordless $|\M|$-cycle. Thus, the path $\underline{\mathfrak{D}}$ must be chordless as well, as the existence of a chord for $\underline{\mathfrak{D}}$ would imply the existence of a chord for $\Cov{k}$. Moreover, because $n<|\M|$, we know that $\underline{\mathfrak{D}}$ is not an improper $3$-cycle, thus we can use the notation of Remark \ref{rem: notation}. Using this notation, we formulate the following proposition.

\begin{proposition}\label{prop: cycle3}
Let $\scenario$ be a cyclic scenario, and let $k\ge2$, $2\leq n<|\M|$. Let $\underline{\mathfrak{D}}:=\{\C_1,\dots \C_n\}\subseteq\Cov{k}$ be an $n$-path for $\Cov{k}$. Then the set 
\[
\underline{\D}':=\{K_1',K_1, K_2,\dots K_{n}\},
\]
where $K_1':=C_1^1$ and $K_i:=C_i^2$ for all $1\leq i\leq n$, is an $(n+1)$-path for $\Cov{k-1}$
\end{proposition}

\begin{proof}

The argument we will use to prove that the $K_i$'s are all distinct essentially coincides with the one of Proposition \ref{prop: cycle2}. Let us start by proving that $K_i\neq K_{i+1}$. Suppose $K_i=K_{i+1}$ for some $1\leq i\leq n-1$. Then we have $C_i^2=C_{i+1}^2$, which implies $\C_{i+1}=\{C_{i+1}^1,C_{i+1}^1\}=\{C_{i+1}^1\}$, which is not a context of $\Cov{k}$. We can prove that $K_1'\neq K$ in the same way. Now, suppose there are two non-consecutive indices such that $K_i=K_j$. This implies $C_i^2=C_j^2$. Hence
\[
C_i\cap C_j=\{C_i^1,C_i^2\}\cap \{C_j^1,C_j^2\}=\{C_i^2\}\neq\emptyset,
\]
however, this would imply that the cover $\Cov{k}$, which is chordless by definition of a cyclic scenario, has a chord $\{C_i, C_j\}\in\Cov{k}$, which is obviously a contradiction. We can prove that $K_1'\neq K_j$ for all $2\leq j\leq n$ in the same way. 

We are only left to prove that consecutive $K_i$'s intersect. This can be done following exactly the same argument as in the proof of Proposition \ref{prop: cycle2}.
\end{proof}

\section{The cohomology of cyclic models}

In this section we will formalise the intuitive idea discussed at the end of Section \ref{sec: interpretation}, and generalise it to prove that we can always find a cohomological witness for contextuality in the joint models of a cyclic empirical model. 

\subsection{Preliminaries}

First of all, we need to introduce some preliminary definitions. Let $\S$ be an empirical model on a measurement scenario $\scenario$. We will choose, as a representative for each joint model $\Model{k}$, the presheaf of abelian groups
\begin{equation}\label{equ: representative}
\FModel{k}:=F_{\mathbb{Z}_2}\Model{k}: \P(\Meas{k})^{op}\longrightarrow \textbf{AbGrp}.
\end{equation}
We can now formulate the following definition, which is a natural extension of Definition \ref{defn: LCk} to account for cohomology. 

\begin{defn}\label{defn: CLCk}
Let $\S$ be an empirical model on a measurement scenario $\scenario$, with representative $\F$ as in \eqref{equ: representative}. Let $s$ be a local section of $\S$. In view of the results of Section \ref{sec: contextuality}, we say that \emph{$\S$ is $\CLCk{k}$ at $s$}, and write $\CLCk{k}(\S, s)$, if we have $\CLC(\Model{k},t)$ for every local section $t$ of $\Model{k}$ such that $s\in\flatten(t)$. 
\end{defn}

We can use this definition to extend Theorem \ref{thm: main} to joint models:

\begin{theorem}\label{thm: main2}
Let $\S$ be an empirical model. Given a section $s$ of $\S$, if there exists a $k\ge0$ such that $\CLCk{k}(\S,s)$, then $\LC(\S,s)$. Moreover, $\CSC(\Model{k})\Rightarrow \SC(\S)$. 
\end{theorem}

\begin{proof}
Suppose $\CLCk{k}(\S, s)$, i.e. $\CLC(\Model{k},t)$ for every local section $t$ of $\Model{k}$ such that $s\in\flatten(t)$. By Theorem \ref{thm: main}, it follows that $\LC(\Model{k}, t)$ for all $t$ such that $s\in\flatten(t)$. In other words, we have $\LCk{k}(\S, s)$ (cf. Definition \ref{defn: LCk}). By Corollary \ref{cor: LCk}, this implies that $\S$ is logically contextual at $s$. 

Now, suppose $\CSC(\Model{k})$, then, by Theorem \ref{thm: main}, we have $\SC(\Model{k})$. By Corollary \ref{cor: SC2} we conclude that $\SC(\S)$. 
\end{proof}

We now introduce the notion of \emph{partial family}.

\begin{defn}\label{defn: path}
Let $\S$ be an empirical model over a measurement scenario $\scenario$, and $n,k\ge 1$. An \emph{
$n$-partial family for $\FModel{k}$} is a family 
\[
\left\{f_\C\in \FModel{k}(\C)\right\}_{\C\in\underline{\mathfrak{D}}}
\]
over an 
$n$-path $\underline{\mathfrak{D}}=\{\C_1,\dots,\C_n\}\subseteq\Cov{k}$, which is compatible for $\FModel{k}$, and satisfies the following conditions:
\begin{align}
f_{\C_1}\mid_{C_1^1} &\in\Model{k-1}(C_1^1), \label{equ: 1}\\
f_{\C_n}\mid_{C_n^2} &\in\Model{k-1}(C_n^2),\label{equ: 2}
\end{align}
(cf. Remark \ref{rem: notation} for notation).
A partial family is called \emph{standard} if there exists a family $\{s_\C\in\Model{k}(\C)\}_{\C\in\underline{\mathfrak{D}}}$, compatible for $\Model{k}$ such that 
\begin{align}
s_{\C_1}\mid_{C_1^1} &= f_{\C_1}\mid_{C_1^1}, \label{equ: extreme1}\\
s_{\C_n}\mid_{C_n^2} &= f_{\C_n}\mid_{C_n^2}. \label{equ: extreme2}
\end{align}
In this case, the family $\{s_\C\in\Model{k}(\C)\}_{\C\in\underline{\mathfrak{D}}}$ is called the \emph{the standard form of $\left\{f_\C\in \FModel{k}(\C)\right\}_{\C\in\underline{\mathfrak{D}}}$}
\end{defn}
Note that a $1$-partial family for $\FModel{k}$ is simply a single section $f\in\FModel{k}(\C)$ over a context $\C\in\Cov{k}$, which verifies conditions \eqref{equ: 1} and \eqref{equ: 2}.

We have introduced partial families in order to model the typical cohomology false positive. Indeed, non-standard partial families are nothing but partial families of $\FModel{k}$ (i.e. families of linear combinations of sections of $\Model{k}$) that cannot be replaced by simple families of $\Model{k}$, just like a cohomology false positive is a family for $\F$ which cannot be replaced by a family of $\S$. We give some graphical intuition on partial families in Figure \ref{fig: partial families}, to clarify this concept. Throughout the rest of the section, we will show how non-standard families can be suppressed by applying the joint model construction a sufficient amount of times.

\begin{figure}[htbp]
\centering
\includegraphics[scale=0.6]{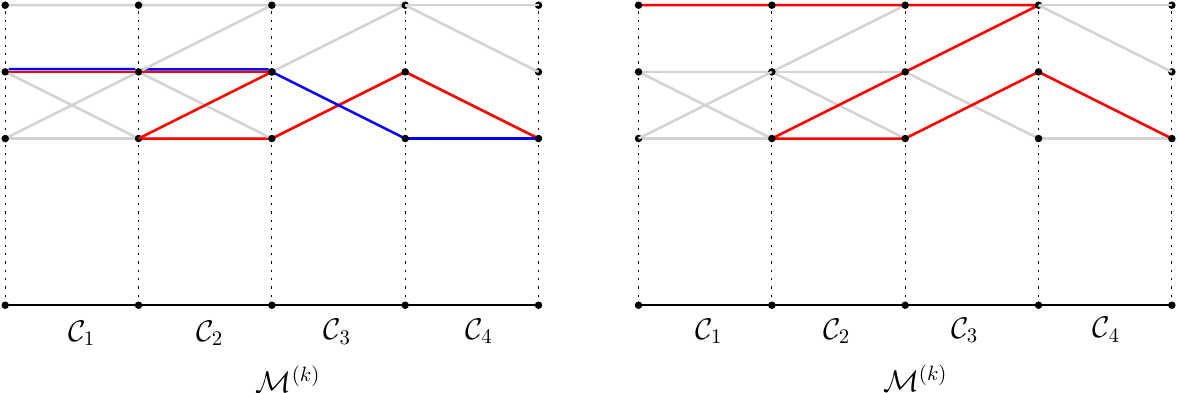}
\caption{Two examples of $4$-partial families for $\FModel{k}$ (in red) over the $4$-path $\{\C_1,\C_2,\C_3,\C_4\}\subseteq \Cov{k}$. On the left, a standard family, with its standard form highlighted in blue. On the right, a non-standard partial family.}\label{fig: partial families}
\end{figure}

\subsection{A complete cohomology invariant for contextuality in cyclic models}

We will now show how to get rid of non-standard partial families. This procedure will require a number of intermediate steps. 

The following lemma is called the no-Z lemma because it formalises the idea, introduced in Section \ref{sec: false positives}, that first joint models do not contain `Z' shaped paths which typically give rise to false positives in cohomology. 

\begin{lemma}[No-Z lemma]\label{lem: no Z}
Let $\S$ be an empirical model over a measurement scenario $\scenario$. Let $k\ge 1$, and $\C=\{C^1,C^2\}\in\Cov{k}$. Every $1$-partial family for $\FModel{k}$ over $\C$ of the form
\begin{equation}\label{equ: Z}
f_\C=(s_1,t_1)+(s_2,t_1)+(s_2,t_2),
\end{equation}
(where $s_i\in\Model{k-1}(C^1)$, $t_i\in\Model{k-1}(C^2)$ for all $i=1,2$) is standard.
\end{lemma}

\begin{proof}
Let $f_\C$ be a $1$-partial family defined by \eqref{equ: Z}. Because both $(s_1,t_1)$ and $(s_2,t_1)$ are in $\Model{k}(\C)$, we know that 
\begin{equation}
s_1\mid_{C^1\cap C^2}=t_1\mid_{C^1\cap C^2}=s_2\mid_{C^1\cap C^2} \label{equ: restriction}
\end{equation}
Moreover, $s_2\mid_{C^1\cap C^2}=t_2\mid_{C^1\cap C^2}$, given that $(s_2,t_2)\in\Model{k}(\C)$. Hence
\[
s_1\mid_{C^1\cap C^2}\stackrel{\eqref{equ: restriction}}{=}s_2\mid_{C^1\cap C^2}=t_2\mid_{C^1\cap C^2}.
\]
Therefore, $(s_1,t_2)\in\Model{k}(\C)$, and we have 
\begin{align*}
(s_1,t_2)\mid_{C^1} &=s_1\mid_{C^1}=s_1\mid_{C^1}+\underbrace{2\cdot s_2\mid_{C^1}}_{=0}=f_\C\mid_{C^1}\\
(s_1,t_2)\mid_{C^2} &=t_2\mid_{C^2}=t_2\mid_{C^2}+\underbrace{2\cdot t_1\mid_{C^2}}_{=0}=f_\C\mid_{C^2},
\end{align*}
which correspond to conditions \eqref{equ: extreme1} and \eqref{equ: extreme2} (we have used the fact that the coefficients are in $\mathbb{Z}_2$, hence $2=0$). This proves that $(s_1,t_2)$ is the standard form of $f_\C$. 
\end{proof}

We will now generalise the no-Z lemma to all the $1$-partial families for $\FModel{k}$. The proof essentially consists of a recursive algorithm which takes a $1$-partial family as input, and outputs a standard form by repeatedly applying the no-Z lemma to the first three segments of the partial family, which -- we show -- are always in a `Z' shape. 

 \begin{lemma}\label{lem: base case}
Let $\S$ be an empirical model on a scenario $\scenario$, and let $k\ge 1$. All the $1$-partial families for $\FModel{k}$ are standard.
\end{lemma}

\begin{proof}
A $1$-partial family is a single section $f_\C\in\FModel{k}(\C)$ over a single context $\C=\{C^1,C^2\}\in\Cov{k}$, which verifies conditions \eqref{equ: 1} and \eqref{equ: 2}. We provide an algorithm that constructs a standard form $s_\C\in\Model{k}(\C)$. Let us start by enumerating the possible sections at $C^1$ and $C^2$ by denoting $\Model{k-1}(C^1)=\{s_1,s_2,\dots, s_n\}$, and $\Model{k-1}(C^2)=\{t_1,t_2,\dots, t_m\}$. Let
\[
I:=\{(i,j)\in[n]\times[m]\mid (s_i,t_j)\in\Model{k}(\C)\},
\]
where $[l]:=\{1,2,\dots, l\}$. By definition of $\FModel{k}$, the section $f_\C$ can be written as a formal linear combination of sections in $\Model{k}(\C)$:
\[
f_\C=\sum_{(i.j)\in I} \alpha_{ij}\cdot (s_i,t_j),
\]
where $\alpha_{ij}\in\mathbb{Z}_2$. 

If $f_\C\in\Model{k}(\C)$, then we are done, as $f_\C$ is already in standard form. 
Otherwise, we know by \eqref{equ: 1} that $f_\C\mid_{C^1}\in\Model{k-1}(C^1)$. Therefore, we can assume w.l.o.g. that $f_\C\mid_{C^1}=s_1$. Because of this, there exists a $j_1\in[m]$ such that $\alpha_{1j_1}=1$, and we can assume w.l.o.g. that $j_1=1$, which means that $f_\C$ contains the section $(s_1,t_1)\in\Model{1}(\C)$ in its summands, i.e.
\begin{equation}\label{equ: fc}
f_\C=(s_1,t_1)+\sum_{(i,j)\neq (1,1)}\alpha_{ij}(s_i,t_j)
\end{equation}

By equation \eqref{equ: 2}, we know that $f_\C\mid_{C^2}\in\Model{k-1}(C^2)$ and we can denote $f_\C\mid_{C^2}=t_l$, for some $l\in[m]$. If $l=1$, then we can immediately return $(s_1,t_1)$ as the standard form of $f_\C$. Otherwise we assume $l\neq 1$. 
\\
\begin{claim}{1}
There exists an index $i_1\neq 1$, $i_1\in[n]$ , such that $(i_1, 1)\in I$ and $\alpha_{i_11}=1$. W.l.o.g. we let $i_1=2$.
\end{claim}
~\\
\begin{claimproof}
Suppose \emph{ab absurdo} $\alpha_{i1}=0$ for all $1\neq i\in[n]$ such that $(i,1)\in I$. Then,  given \eqref{equ: fc}, we have
\[
f_\C=(s_1,t_1)+\sum_{j\neq 1}\alpha_{ij}(s_i,t_j).
\]
This implies
\[
f_\C\mid_{C^2}=t_1+\sum_{j\neq 1}\alpha_{ij}t_j,
\]
which always contains the summand $t_1\neq t_l$, and thus can never equal $t_l$, which is a contradiction. 
\end{claimproof}
~\\

Claim 1 shows that $f_\C$ always contains the summand $(s_2,t_1)$, i.e.
\begin{equation}\label{equ: fc2}
f_\C=(s_1,t_1)+(s_2,t_1)+\sum_{\substack{(i,j)\neq (1,1)\\ (i,j)\neq(2,1)}}\alpha_{ij}(s_i,t_j).
\end{equation}

\begin{claim}{2}
There exists an index $j_2\in[m]$, $j_2\neq 1$, such that $(2,j_2)\in I$ and $\alpha_{2j_2}=1$. W.l.o.g. we let $j_2=2$.
\end{claim}
~\\
\begin{claimproof}
Suppose by contradiction that $\alpha_{2j}=0$ for all $1\neq j\in[m]$ such that $(2,j)\in I$. Then, given \eqref{equ: fc2}, we have
\[
f_\C=(s_1,t_1)+(s_2,t_1)+\sum_{\substack{(i,j)\neq (1,1) \\ i\neq 2}}\alpha_{ij}(s_i,t_j).
\]
This implies 
\[
f_\C\mid_{\C^1}=s_1+s_2+\sum_{\substack{(i,j)\neq (1,1) \\ i\neq 2}}\alpha_{ij}s_i,
\]
which always contains the summand $s_2$, and thus can never equal $s_1$, which is a contradiction. 
\end{claimproof}
~\\

Claim 2 shows that $f_\C$ always contains the summand $(s_2,t_2)$, i.e. 
\[
f_\C=(s_1,t_1)+(s_2,t_1)+(s_2,t_2)+\sum_{\substack{(i,j)\neq (1,1)\\ (i,j)\neq(2,1)\\ (i,j)\neq (2,2)}}\alpha_{ij}(s_i,t_j).
\]
Notice how the first three summands are exactly the same as in \eqref{equ: Z}. This means that these first three `steps' of the partial family $f_\C$ are in a `Z' shape. This allows us to apply the no-Z lemma and substitute the $Z$ with a section in $\Model{k}(\C)$, as shown in Figure \ref{fig: Z}. 

\begin{figure}[htbp]
\centering
\includegraphics[scale=0.6]{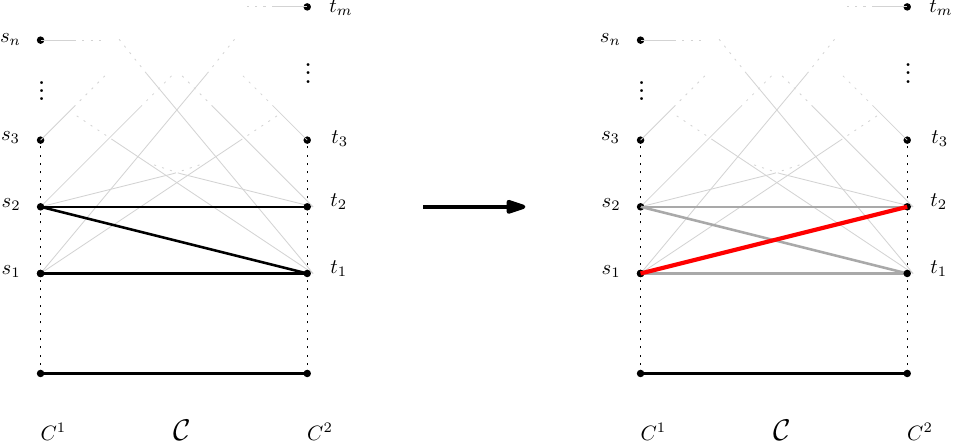}
\caption{A visualisation of the proof. On the left hand side the `Z' shape at the beginning of the partial family. On the right hand side, the substitution of the `Z' with a section of $\Model{k}(\C)$.}\label{fig: Z}
\end{figure}

More formally, by the no-Z lemma (Lemma \ref{lem: no Z}), $(1,2)$ must be in $I$, and the section $(s_1,t_2)\in\Model{k}(\C)$ is the standard form of the partial family $(s_1,t_1)+(s_2,t_1)+(s_2,t_2)$.

If $l=2$, then $(s_1,t_2)$ is the standard form of $f_\C$ and we are done. Otherwise we can re-input the partial family
\[
f'_\C:=(s_1,t_2)+\sum_{\substack{(i,j)\neq (1,1)\\ (i,j)\neq(2,1)\\ (i,j)\neq (2,2)}}\alpha_{ij}(s_i,t_j)
\]
into the algorithm. The algorithm obviously terminates as there is only a finite amount of sections. 
\end{proof}

The following theorem is the key result of the paper. It shows that, on cyclic scenarios, all the $n$-partial families for $\FModel{k}$, where $n\leq k$, can be replaced by a standard form of the same size. In other words, any potential cohomological false positive of size $n\leq k$ can be erased. This fact will lead us to a fundamental result, namely that it is sufficient to take the $(|\M|-1)$-joint model of a cyclic scenario to remove every cohomology false positive with certainty (Corollary \ref{cor: main}). 

\begin{theorem}\label{thm: standard families}
Let $\S$ be an empirical model on a cyclic scenario $\scenario$. Let $k\ge 1$ and let $n$ be such that $n\leq k$ and $n<|\M|$. All the $n$-partial families for $\FModel{k}$ are standard.
\end{theorem}

\begin{proof}
We will proceed by induction on $k$. The base case $k=1$ has been proved in Lemma \ref{lem: base case}. We will now suppose $k\ge 2$.

Let $P:=\left\{f_\C\in \FModel{k}(\C)\right\}_{\C\in\mathfrak{D}}$ be an $n$-partial family for $\FModel{k}$ over the $n$-path $\underline{\mathfrak{D}}=\{\C_1,\dots,\C_n\}\subseteq\Cov{k}$. Notice that, for $n=1$, the result follows directly from Lemma \ref{lem: base case}. 

Suppose $n\ge 2$. Because $\scenario$ is cyclic, by Proposition \ref{prop: cyclicn} we know that $\Cov{k}$ is a chordless $|\M|$-cycle. Moreover, $\underline{\mathfrak{D}}$ cannot be an improper $3$-cycle, because $n<|\M|$. Thus we can use the notation of Remark \ref{rem: notation}. Let 
\[
\underline{\D}':=\{K_1,K_2,\dots K_{n-1}\}=\{C_1^2, C_2^2, \dots C_{n-1}^2\}\subseteq\Cov{k-1}.
\]
Then $\underline{\D}'$ is an $(n-1)$-path for $\Cov{k-1}$. Indeed, by Proposition \ref{prop: cycle3} we know that $\underline{\D}'\cup\{C_1^1, C_n^2\}$ is an $(n+1)$-path for $\Cov{k-1}$. 

Define a family $P':=\left\{t_{K_i}\in \FModel{k-1}(K_i)\right\}_{i=1}^{n-1}$ by 
\begin{equation}\label{equ: t}
t_{K_i}:=f_{\C_i}\mid_{K_i},
\end{equation}
(cf. Remark \ref{rem: subtlety}). 
~\\
\begin{claim}{1}
The family $P'$ is an $(n-1)$-partial family for $\FModel{k-1}$. 
\end{claim}
~\\
\begin{claimproof}
Let us start by proving that $P'$ is compatible for $\FModel{k-1}$. Let $1\leq i\leq n-2$. We have
\[
\begin{split}
t_{K_i}\mid_{K_i\cap K_{i+1}} & = \left(f_{\C_i}\mid_{K_i}\right)\mid_{K_i\cap K_{i+1}}\stackrel{(\ast)}{=}\left(f_{\C_i}\mid_{\C_i\cap \C_{i+1}}\right)\mid_{K_i\cap K_{i+1}}\\
&\stackrel{(\dagger)}{=}\left(f_{\C_{i+1}}\mid_{\C_i\cap \C_{i+1}}\right)\mid_{K_i\cap K_{i+1}}\stackrel{(\ast)}{=}\left(f_{\C_{i+1}}\mid_{K_i}\right)\mid_{K_i\cap K_{i+1}}\\
&=f_{\C_{i+1}}\mid_{K_i\cap K_{i+1}}=\left(f_{\C_{i+1}}\mid_{K_{i+1}}\right)\mid_{K_i\cap K_{i+1}}\\
&=t_{K_{i+1}}\mid_{K_i\cap K_{i+1}},
\end{split}
\]
where we have used the fact that $\C_i\cap \C_{i+1}=\{C_i^2\}=\{K_i\}$ in the equalities $(\ast)$ (cf. Remark \ref{rem: subtlety}), and the fact that $P$ is compatible for $\FModel{k}$ in equality $(\dagger)$. With the usual notation $K_i:=\{k_i^1, k_i^2\}$, since $\{k_1^1,k_1^2,k_2^2,\dots, k_{n-1}^2\}$ is an $n$-path for $\Cov{k-2}$ by Proposition \ref{prop: cycle3}, we know that $k_1^1$ is the first vertex of the path, which means that $\{k_1^1\}=C_1^1\cap C_1^2$. In view of Remark \ref{rem: subtlety}, we have 
\[
\begin{split}
t_{K_1}\mid_{k_1^1} &=\left(f_{\C_1}\mid_{K_1}\right)\mid_{k_1^1}=\left(f_{\C_1}\mid_{C_1^2}\right)\mid_{k_1^1}= \left(f_{\C_1}\mid_{C_1^2}\right)\mid_{C_1^1\cap C_1^2}=f_{\C_1}\mid_{C_1^1\cap C_1^2}\\
&=\left(f_{\C_1}\mid_{C_1^1}\right)\mid_{C_1^1\cap C_1^2}=\left(f_{\C_1}\mid_{C_1^1}\right)\mid_{k_1^1}
\end{split}
\]
Because $f_{\C_1}\mid_{C_1^1}\in\Model{k-1}(C_1^1)$ by condition \eqref{equ: 1}, we must have
\begin{equation}\label{equ: t1}
t_{K_1}\mid_{k_1^1} =\left(f_{\C_1}\mid_{C_1^1}\right)\mid_{C_1^1\cap C_1^2}\in\Model{k-1}(k_1^1),
\end{equation}
hence $P'$ satisfies \eqref{equ: 1}.

To prove \eqref{equ: 2}, we start by a simple observation, namely that, because $K_{n-1}=C_{n-1}^2$, we have $\{K_{n-1}\}=\C_{n-1}\cap \C_n$. Therefore
\begin{equation}\label{equ: observation}
t_{K_{n-1}}=f_{\C_{n-1}}\mid_{K_{n-1}}=f_{\C_{n-1}}\mid_{C_{n-1}\cap C_n}=f_{\C_n}\mid_{C_{n-1}\cap C_n}=f_{\C_n}\mid_{K_{n-1}},
\end{equation}
where we have used the fact that $P$ is compatible in the third equality. Now, with a similar argument as before, given that $\{k_{n-1}^2\}=C_{n-1}^2\cap C_n^2$, we have
\[
\begin{split}
t_{K_{n-1}}\mid_{k_{n-1}^2} &\stackrel{\eqref{equ: observation}}{=}\left(f_{\C_n}\mid_{K_{n-1}}\right)\mid_{k_{n-1}^2}=\left(f_{\C_n}\mid_{C_{n-1}^2}\right)\mid_{k_{n-1}^2}= \left(f_{\C_n}\mid_{C_{n-1}^2}\right)\mid_{C_{n-1}^2\cap C_n^2}\\
&\stackrel{\phantom{\eqref{equ: observation}}}{=}f_{\C_n}\mid_{C_{n-1}^2\cap C_n^2}=\left(f_{\C_n}\mid_{C_n^2}\right)\mid_{C_{n-1}^2\cap C_n^2}=\left(f_{\C_n}\mid_{C_n^2}\right)\mid_{k_{n-1}^2}.
\end{split}
\]
Because $f_{\C_n}\mid_{C_n^2}\in\Model{k-1}(C_1^1)$ by condition \eqref{equ: 2}, we must have
\begin{equation}\label{equ: tn-1}
t_{K_{n-1}}\mid_{k_{n-1}^2} =\left(f_{\C_n}\mid_{C_n^2}\right)\mid_{C_{n-1}^2\cap C_n^2}\in\Model{k-1}(k_{n-1}^2),
\end{equation}
which means that $P'$ satisfies \eqref{equ: 2}.
\end{claimproof}
~\\

Because $P'$ is an $(n-1)$-partial family for $\FModel{k-1}$, by inductive hypothesis, we know that $P'$ is standard. Let $S:=\{s_{K_i}\in\Model{k-1}(K_i)\}_{i=1}^{n-1}$ be the standard form of $P'$, i.e. $S$ is compatible for $\Model{k-1}$ and it is such that
\begin{align}
s_{K_1}\mid_{k_1^1} &= t_{K_1}\mid_{k_1^1}, \label{equ: st1}\\ 
s_{K_{n-1}}\mid_{k_{n-1}^2} &= t_{K_{n-1}}\mid_{k_{n-1}^2}. \label{equ: st2}
\end{align}
Consider the family $G:=\{g_{\C_i}\in\Model{k}(\C_i)\}_{i=1}^n$
where 
\[
g_{\C_i}:=
\begin{cases}
\left(f_{\C_1\mid_{C_1^1}},s_{K_1}\right) & \text{if } i=1,\\
\left(s_{K_{n-1}}, f_{\C_n}\mid_{C_n^2}\right) & \text{if } i=n,\\
(s_{K_{i-1}},s_{K_i}) & \text{for all } 2\leq i\leq n-1.
\end{cases}
\]
\begin{claim}{2}
The family $G$ is a standard form for $P$. 
\end{claim}
~\\
\begin{claimproof}
First of all, we need to check that $g_{\C_i}$ is indeed an element of $\Model{k}(\C_i)$ for all $1\leq i\leq n$. We have
\[
\left(f_{\C_1}\mid_{C_1^1}\right)\mid_{C_1^1\cap K_1}=\left(f_{\C_1}\mid_{C_1^1}\right)\mid_{C_1^1\cap C_1^2}\stackrel{\eqref{equ: t1}}{=}t_{K_1}\mid_{k_1^1}\stackrel{\eqref{equ: st1}}{=}s_{K_1}\mid_{k_1^1}=
s_{K_1}\mid_{C_1^1\cap K_1}.
\]
Similarly, 
\[
\begin{split}
s_{K_{n-1}}\mid_{K_{n-1}\cap C_n^2}&=s_{K_{n-1}}\mid_{k_{n-1}^2}\stackrel{\eqref{equ: st2}}{=}t_{K_{n-1}}\mid_{k_{n-1}^2}\stackrel{\eqref{equ: tn-1}}{=}\left(f_{\C_n}\mid_{C_n^2}\right)\mid_{C_{n-1}^2\cap C_n^2}\\
&=\left(f_{\C_n}\mid_{C_n^2}\right)\mid_{K_{n-1}\cap C_n^2}.
\end{split}
\]
Finally, let $2\leq i\leq n-1$. We readily have
\[
s_{K_{i-1}}\mid_{K_{i-1}\cap K_i}=s_{K_i}\mid_{K_{i-1}\cap K_i}
\]
by the simple fact that $S$ is compatible for $\Model{k-1}$. 

The fact that $G$ satisfies equations $\eqref{equ: extreme1}$ and $\eqref{equ: extreme2}$ for $P$ trivially follows from the very definition of $G$, indeed
\begin{align*}
g_{\C_1}\mid_{\C_1^1}&=\left(f_{\C_1\mid_{C_1^1}},s_{K_1}\right)\mid_{C_1^1}=f_{\C_1}\mid_{C_1^1};\\
g_{\C_n}\mid_{\C_n^2}&=\left(s_{K_{n-1}}, f_{\C_n}\mid_{C_n^2}\right)\mid_{C_n^2}=f_{\C_n}\mid_{C_n^2}.
\end{align*}
\end{claimproof}
~\\
Thanks to this claim, we have successfully proved that $P$ is standard. 
\end{proof}

We can now introduce a complete cohomology characterisation of logical and strong contextuality for cyclic scenarios:

\begin{theorem}\label{cor: main}
Let $\S$ be an empirical model on a cyclic scenario $\scenario$. Let $C\in\M$ and $s\in\S(C)$. Then we have
\[
\LC(\S, s) \Leftrightarrow \CLCk{n}(\S, s),
\]
where $n:=|\M|-1$.
Moreover, 
\[
\SC(\S)  \Leftrightarrow \CSC(\Model{n})
\]
\end{theorem}
\begin{proof}
The implications $\CLCk{n}(\S, s)\Rightarrow\LC(\S, s)$ and $\CSC(\Model{n})\Rightarrow \SC(\S)$ have already been proven in Theorem \ref{thm: main2}. 


To prove the converse, we will show that $\neg \CLCk{n}(\S, s)  \Rightarrow \neg\LC(\S, s)$. Suppose $\neg \CLCk{n}(\S, s)$. By Definition \ref{defn: CLCk}, this implies that there exists a context $\C_0\in\Cov{n}$ and a section $t\in\Model{n}(\C_0)$ such that $s\in\flatten(t)$ and $\neg\CLC(\Model{n}, t)$. Thus, there exists a compatible family 
\[
F:=\left\{f_\C\in\FModel{n}(\C)\right\}_{\C\in\Cov{n}}
\]
such that $f_{\C_0}=t$. Because $\scenario$ is cyclic, we know by Proposition \ref{prop: cyclicn} that $\Cov{n-1}$ is also cyclic. Theferefore, 
\[
\Cov{n}=\{\C_0,\C_1,\dots, \C_{n}\}
\]
is a chordless $|\M|$-cycle, which implies that $\{\C_1,\dots, \C_{n}\}$ is a chordless $n$-path (in edge representation) for $\Cov{n}$. 
Let
\[
P:=\left\{f_{\C_i}\in\FModel{n}(\C_i)\right\}_{i=1}^n.
\]
Then $P$ is a $n$-partial family for $\FModel{n}$, indeed it is compatible because $F$ is compatible, and we have, with the usual notation
\[
f_{\C_1}\mid_{C_1^1}=f_{\C_1}\mid_{\C_0\cap\C_1}\stackrel{(\ast)}{=}f_{\C_0}\mid_{\C_0\cap \C_1}=t\mid_{\C_0\cap C_1}\in\Model{n-1}(\C_0\cap \C_1),
\]
and 
\[
f_{\C_{n}}\mid_{C_{n}^2}=f_{\C_{n}}\mid_{\C_{n}\cap \C_0}\stackrel{(\ast)}{=}f_{\C_0}\mid_{\C_{n}\cap \C_0}
=t\mid_{\C_{n}\cap \C_0}\in\Model{n-1}(\C_{n}\cap \C_0),
\]
where we have used the fact that $F$ is compatible in equalities $(\ast)$.
By Theorem \ref{thm: standard families}, we know that $P$ is standard. Thus there exists a family 
\[
P':=\{s_{\C_i}\in\Model{n}(\C_i)\}_{i=1}^n
\]
such that
\[
\begin{split}
s_{\C_1}\mid_{C_1^1} &= f_{\C_1}\mid_{C_1^1}=t\mid_{\C_0\cap C_1},\\
s_{\C_n}\mid_{C_n^2} &= f_{\C_n}\mid_{C_n^2}=t\mid_{\C_{n}\cap \C_0}.
\end{split}
\]
Therefore, the family 
\[
P'\cup\{t\}=\{s_{\C_i}\in\Model{n}(\C_i)\}_{i=1}^n\cup\{t\}
\]
is a compatible family for $\Model{n}$ that contains $t$. Thus we have $\neg \LC(\Model{n},t)$, which means that $\neg\LCk{n}(\S,s)$. It follows from Corollary \ref{cor: LCk} that $\S$ is \emph{not} logically contextual at $s$. 

Suppose now $\neg\CSC(\Model{n})$. Then there exists a section $t$ of $\Model{n}$ such that $\neg\LC(\Model{n}, t)$. Consider an arbitrary section $s$ of $\S$ such that $s\in\flatten(t)$ (such a section always exists by definition of the model $\Model{k}$. Then we have $\neg\LCk{n}(\S,s)$, and we can apply the same argument used before to show that this implies $\neg\LC(\S,s)$, which in turn implies $\neg\SC(\S)$. 
\end{proof}

This theorem tells us that if we want to study the contextuality of a cyclic scenario $\scenario$, it is sufficient to analyse the cohomology of its $(|\M|-1)$-th joint model to assert with certainty which sections give rise to contextual behavior. This is a major step forward, as this method allows us to get rid of all the false positives introduced in section \ref{sec: false positives}, and many others, as we shall see in the rest of the paper.

\subsubsection{Examples}
In this section we will show how this method applies to the well known false positives that have appeared in the litterature, including the ones we have discussed in Section \ref{sec: false positives}.

\paragraph{The model of Table \ref{tab: model}}
We have already shown that the model $\S$ presented in Table \ref{tab: model} displays a cohomology false positive for the section $s_2$ (cf. Table \ref{tab: enum2}), and suggested that it vanishes as soon as we consider its first joint model $\Model{1}$. We can give a formal proof that this is true. In Figure \ref{fig: falsenegative1} we present once again the bundle diagram of the first joint model, where we introduce variables $a,b,c,d,e,f,g,h,i,j\in\mathbb{Z}_2$ that represent the coefficients to give to every section of $\Model{1}$ in order to construct a compatible family for $\FModel{1}$.

\begin{figure}[htbp]
\centering
\includegraphics[scale=0.6]{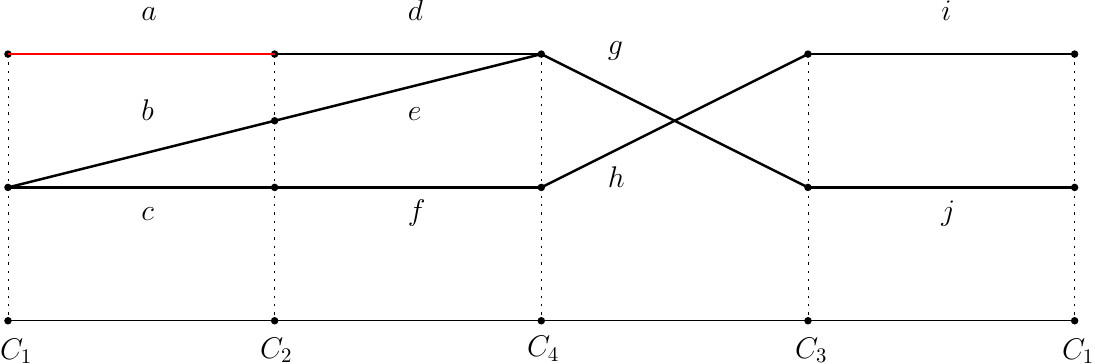}
\caption{The bundle diagram of $\Model{1}$ with the variables in $\mathbb{Z}_2$ corresponding to each section. The seciton $s_2$, responsible for logical contextuality (cf. Section \ref{sec: false positives}) is highlighted in red.}\label{fig: falsenegative1}
\end{figure}
The compatibility conditions of a presumed compatible family for $\FModel{1}$ can be summarised in the following equations:
\[
\begin{aligned}[c]
a &=d,\\
b &=e,\\
c &=f,
\end{aligned}
\qquad\qquad
\begin{aligned}[c]
d\oplus e &= g,\\
f &=h,
\end{aligned}
\qquad\qquad
\begin{aligned}[c]
h &=i,\\
g &=j,
\end{aligned}
\qquad\qquad
\begin{aligned}[c]
i &=a,\\
j &=b\oplus c.
\end{aligned}
\]
Because the family must contain $s_2$, which is marked in red in Figure \ref{fig: falsenegative1}, we must have $a=1$ and $b=c=0$. It follows directly that $e=f=h=i=0$ and that $d=g=j=1$. However, since $j=b\oplus c$, this leads to $1=0\oplus 0 = 0$, which is obviously a contradiction. We have just proved that the cohomology of $\Model{1}$ does detect the logical contextuality of $\S$ at $s_2$. 

Note that in this case, although $|\M|=4$, it was not necessary to take the third joint model of $\S$ to remove the false positive, as suggested by Theorem \ref{cor: main}. Indeed, the bound $|\M|-1$ is the one that gives us absolute certainty about the non-existence of a false positive. However, as we have just shown, it might be sufficient to take a lower level joint model to remove any false positive from the model.

\paragraph{The Hardy model}
The Hardy model (cf. Table \ref{tab: Hardy}) is perhaps the most well-studied example of cohomological false positive for contextuality \cite{Abramsky3, Abramsky2, Caru}. We have already illustrated its bundle diagram in Figure \ref{fig: Hardy}, and showed in Figure \ref{fig: FJM Hardy} that its first joint model still results in a cohomological false positive for the section $s_1$, at which the Hardy model $\S$ is logically contextual. In Figure \ref{fig: SJM Hardy}, we present  the bundle diagram of the second joint model.

\begin{figure}[htbp]
\centering
\includegraphics[scale=0.55]{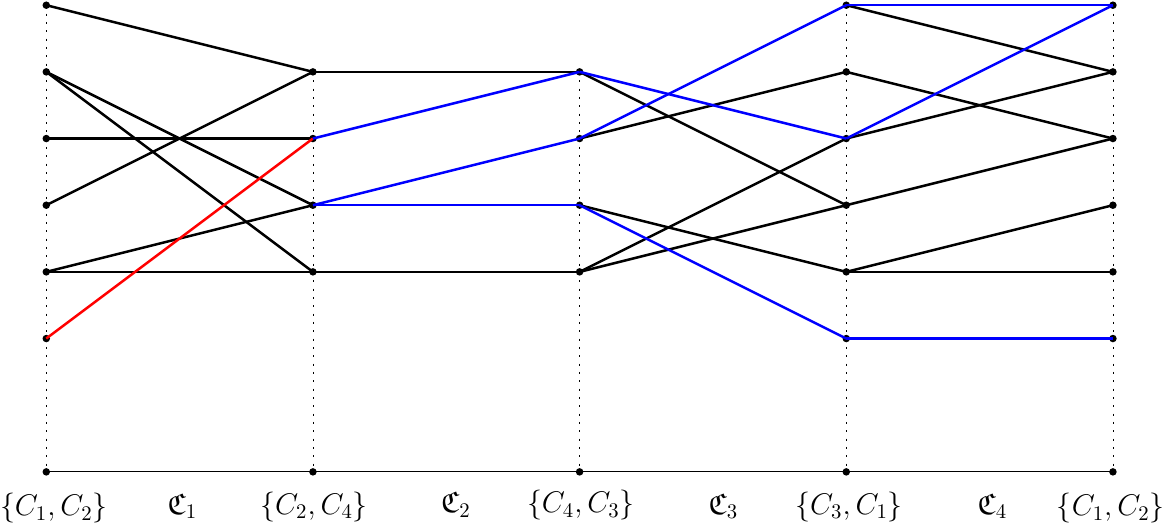}
\caption{The bundle diagram of the second joint model of the Hardy model. The red section is the only section of $\Model{2}$ that contains the original section $s_1$. In blue, a false positive for the red section}\label{fig: SJM Hardy}
\end{figure}

Notice that, even in this case, we still have a compatible family for $\FModel{2}$ containing the only section of $\Model{2}$ that contains $s_1$. Thus, we must consider the third joint model to get rid of the false positive. The third joint model of the Hardy model is presented in Figure \ref{fig: TJM Hardy}, where we have highlighted in red the only section containing $s_1$. 

\begin{figure}[htbp]
\centering
\includegraphics[scale=0.55]{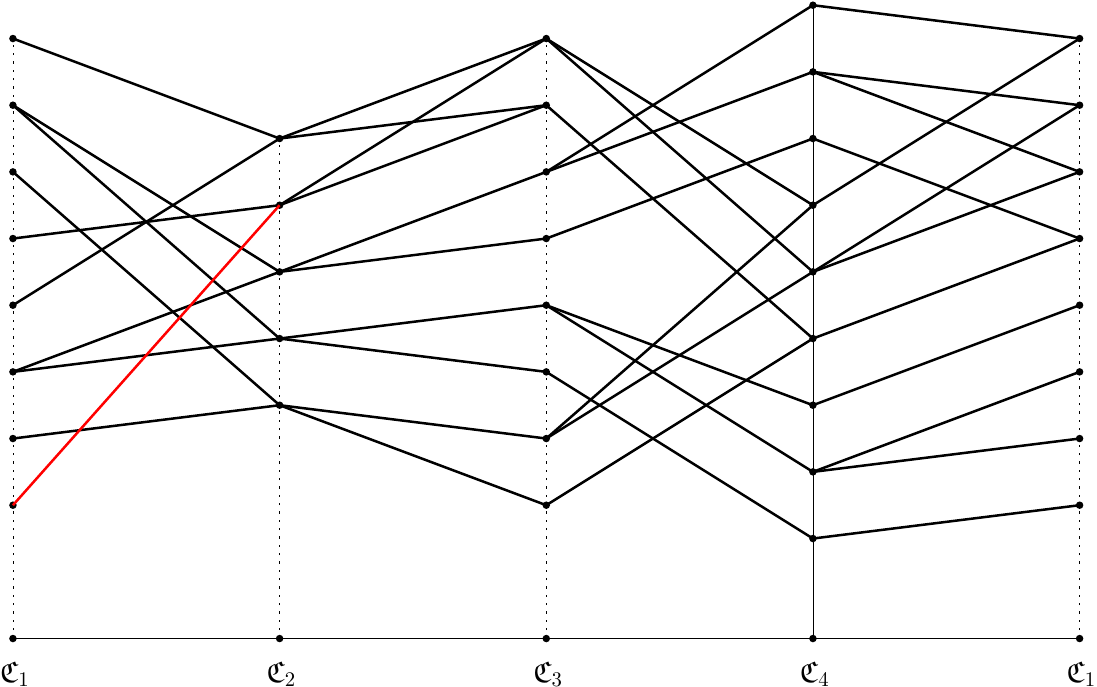}
\caption{The bundle diagram of the third joint model of the Hardy model. The section marked in red is the only section containing $s_1$.}\label{fig: TJM Hardy}
\end{figure}

Since $|\M|=4$, Theorem \ref{cor: main} assures that cohomology does detect the contextuality at the red section. This can be graphically checked by highlighting all the possible attempts to extend the red section to a compatible family for $\FModel{4}$, as shown in Figure \ref{fig: TJMHardy2}.

\begin{figure}[htbp]
\centering
\includegraphics[scale=0.55]{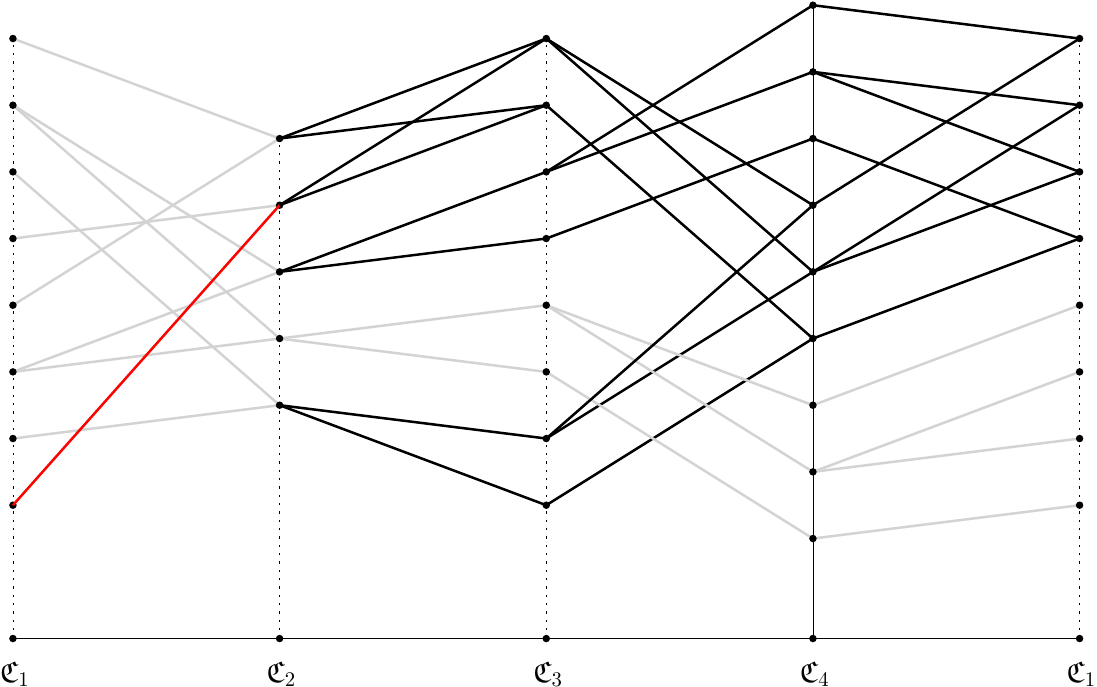}
\caption{It is impossible to extend the red section to a compatible family for $\FModel{4}$.}\label{fig: TJMHardy2}
\end{figure}

This graphical proof can easily be converted into a formal proof following the same idea as in the previous paragraph.
Note that the fact that we had to consider the third joint model of the Hardy model in order to get rid of the false positive shows that the bound $|\M|-1$ of Theorem \ref{thm: main} is tight.

\paragraph{The False positive of \cite{Caru}}

As mentioned before, the false positive of \cite{Caru} (cf. Figure \ref{fig: Gio}) is particularly intereseting because it concerns all the sections of the model. Indeed, the model is cohomologically non-contextual despite being strongly contextual. In other words, there is a cohomology false positive for every single section of the model. In Figure \ref{fig: FJMGio}, we depict the bundle diagram of the first joint model.

\begin{figure}[htbp]
\centering
\includegraphics[scale=0.55]{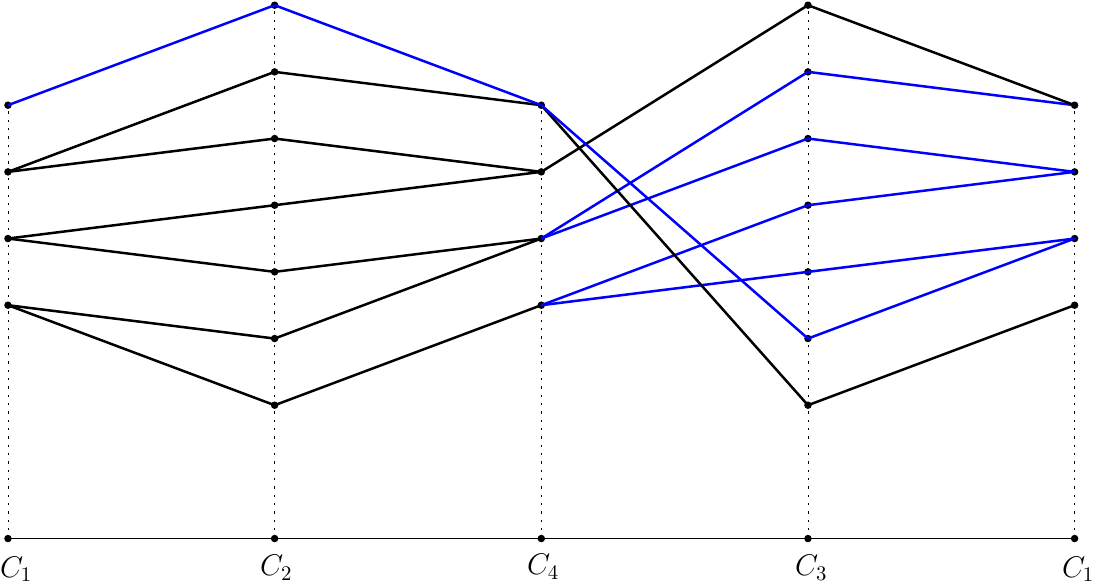}
\caption{The first joint model $\Model{1}$ of the false positive from \cite{Caru}. In blue, a compatible family for $\FModel{1}$.}\label{fig: FJMGio}
\end{figure}

Notice that, for each section, it is still possible to find compatible families for $\FModel{1}$ that contains it, giving rise to a false positive. For example, we have highlighted one such compatible family in blue, which constitutes a false positive for the contextuality of the top section for the context $\{C_1, C_2\}$. 

Therefore, we need to consider the second joint model, whose bundle diagram is depicted in Figure \ref{fig: SJMGio}

\begin{figure}[htbp]
\centering
\includegraphics[scale=0.55]{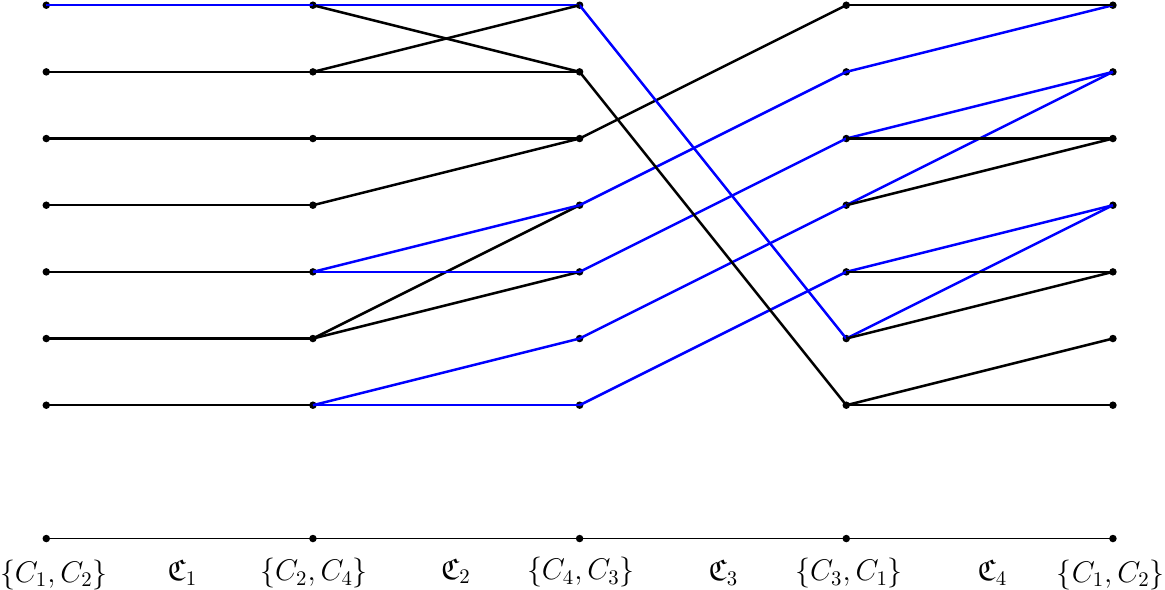}
\caption{The second joint model of the false positive from \cite{Gio}. In blue, a compatible family for $\FModel{2}$.}\label{fig: SJMGio}
\end{figure}

Even in this case, it is still possible to find a cohomology false positive for each section of the model (see e.g. the blue loop highlighted in Figure \ref{fig: SJMGio}). 

In Figure \ref{fig: TJMGio} the bundle diagram of the third joint model is shown. 

\begin{figure}[htbp]
\centering
\includegraphics[scale=0.5]{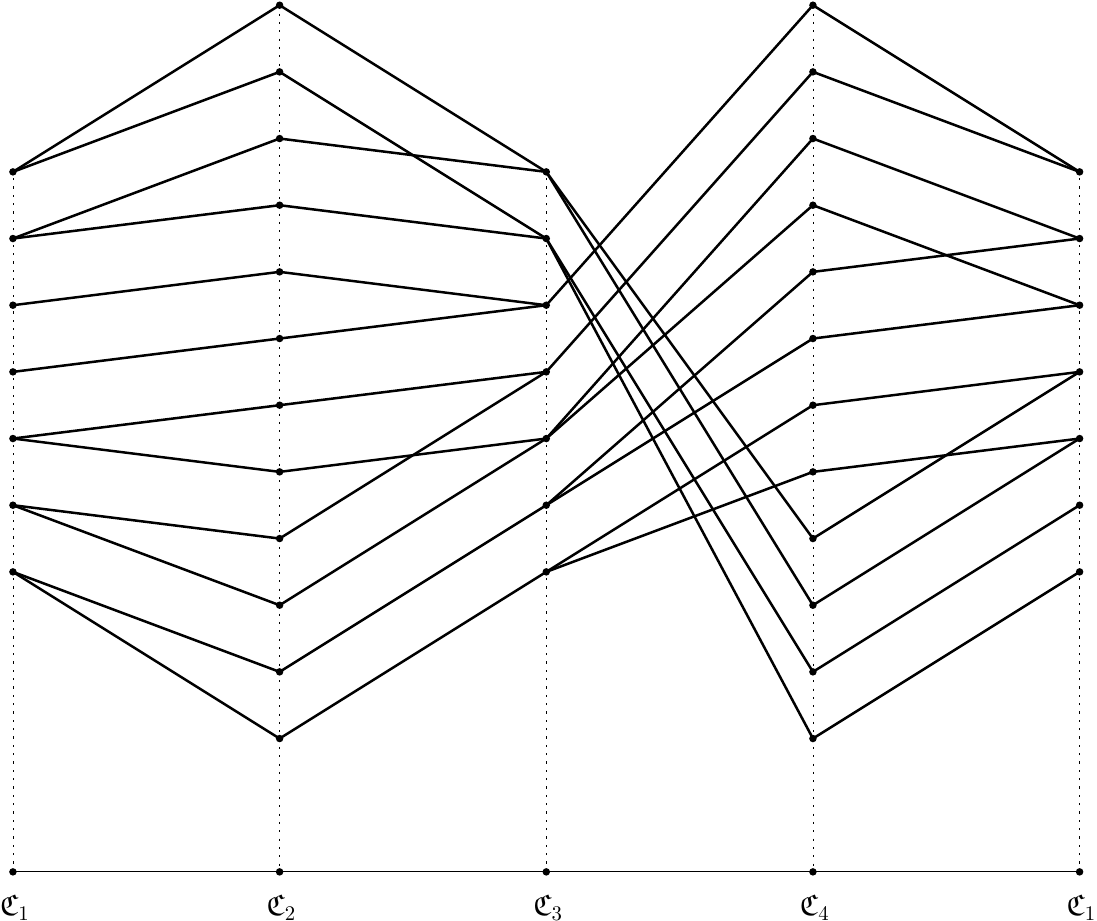}
\caption{The third joint model of the false positive from \cite{Caru}}\label{fig: TJMGio}
\end{figure}

Once again, because $|\M|=4$, we know by Theorem \ref{cor: main}, that cohomology detects contextuality at every section of the model. This can be checked graphically. For example, in Figure \ref{fig: TJMGio2} we show that it is never possible to extend the section marked in red to a compatible family for $\FModel{3}$. The reader can verify that this is true for any section of the model. 

\begin{figure}[htbp]
\centering
\includegraphics[scale=0.5]{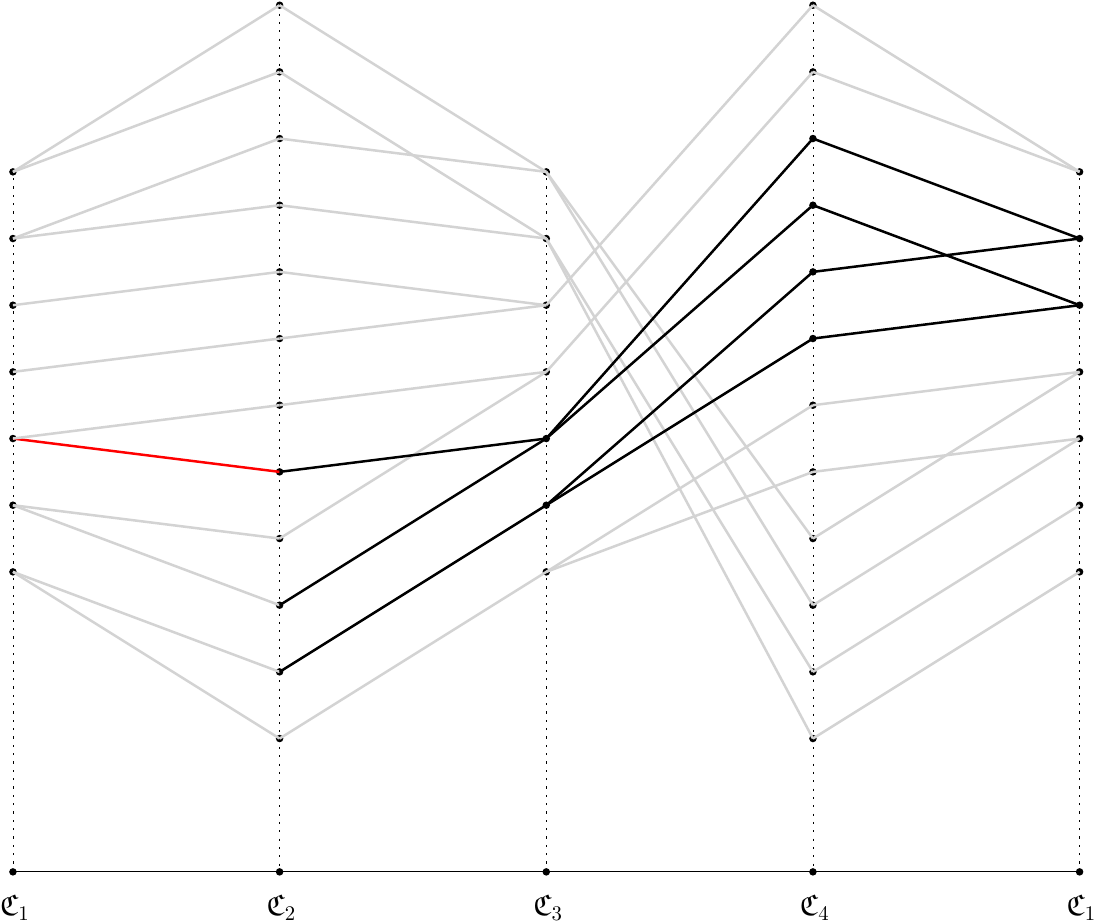}
\caption{A graphical proof of the fact that cohomology of $\FModel{3}$ does detect contextuality at the section marked in red. This is true for any section of the model.}\label{fig: TJMGio2}
\end{figure}

\section{Extending the invariant to general models}
 
In the previous section, we have successfully defined a full cohomology invariant for contextuality for all cyclic models. The goal of this section is to extend this result to arbitrary models. In particular, will show that the invariant can be extended to an extremely vast class of scenarios. This result will lead us to conjecture that the invariant works universally. 

 
Although cyclic models constitute only a fraction of all the possible empirical models, they play a crucial role in the study of contextuality. Indeed, it was proven in \cite{Rui}, thanks to an adaptation of Vorob'ev's theorem \cite{Vorobev}, that a necessary condition for contextuality is the \emph{ciclicity} of the measurement cover. Ciclicity in \cite{Rui} is a notion coming from database theory, and it is not strictly equivalent to the one we introduced in this paper. However, it is easy to prove that any cyclic cover in the sense of \cite{Rui} must contain at least one cyclic subcover in the sense defined here. The fact that the existence of cycles in the cover is necessary for contextuality suggests that contextual features can be observed by focusing uniquely on the cycles. 

To convey this idea, we introduce here the notion of \emph{cyclic contextuality property (CCP)}. The contextual features of models equipped with the CCP can always be recovered by looking at cycles in the cover $\Cov{1}$. 

\begin{defn}
Let $\S$ be an empirical model on a measurement scenario $\scenario$. We say that $\S$ has the \emph{cyclic contextuality property (CCP)} if, for each local section $s$ of $\S$ such that $\LC(\S,s)$, there exists a cycle $\D_\bullet\subseteq\M$ for $\Cov{1}$ (called the \emph{contextual cycle of $s$}) such that $\LC(\S\mid_{\D_\bullet},s)$, where $\S\mid_{\D_\bullet}$ is the model obtained by restricting $\S$ on the subcover $\D_\bullet$. 
\end{defn}

Most empirical models satisfy the CCP. To give an idea of how common this property is, it is sufficient to say that all the models that have appeared in the literature on the sheaf description of contextuality share this property. For instance, it was proven in \cite{AvNTriple} that strong contextuality in quantum models obtained by applying Pauli measurements to stabiliser states can be witnessed by only looking at a cycle of size $4$. 

Remarkably, the cohomology invariant introduced in the previous sections can be immediately extended to all models equipped with the CCP. 

\begin{proposition}
Let $\S$ be a model on a general scenario $\scenario$, and suppose $\S$ has the CCP. For all sections $s$ of $\S$, we have
\[
\LC(\S,s)\Leftrightarrow \CLCk{n-1}(\S,s),
\]
where $n$ denotes the size of the contextual cycle of $s$.
\end{proposition}

\begin{proof}
The implication $\CLCk{n-1}(\S,s)\Rightarrow\LC(\S,s)$ follows from Theorem \ref{thm: main2}. Now, suppose $\LC(\S,s)$. Let $\D_\bullet\subseteq\M$ be the contextual cycle of $s$. By definition, we have $\LC(\S\mid_{\D_\bullet},s)$. The model $\S\mid_{\D_\bullet}$ is defined on the cyclic scenario $\D_\bullet$, thus we can apply Theorem \ref{cor: main} to conclude that $\CLCk{n-1}(\S\mid_{\D_\bullet},s)$, which readily implies $\CLCk{n-1}(\S,s)$. 
\end{proof}

Thanks to this simple proposition, we can extend Theorem \ref{cor: main} to models satisfying the CCP over general scenarios. To prove this, we will need the following proposition

\begin{proposition}
Let $\S$ be an empirical model on a scenario $\scenario$. For all $k\ge 0$ and every section $s$ of $\S$, we have
\[
\CLCk{k}(\S,s)\Rightarrow \CLCk{l}(\S,s) ~\forall l\ge k.
\] 
Similarly,
\[
\CSC(\Model{k})\Rightarrow \CSC(\Model{l})~\forall l\ge k.
\]
\end{proposition}

\begin{proof}
We are going to prove that $\neg\CLCk{k+1}(\S,s)\Rightarrow\neg\CLCk{k}(\S,s)$, and the result will follow by induction. Suppose $\neg\CLCk{k+1}(\S,s)$. Then there exists a context $\C=\{C_1,C_2\}\in\Cov{k+1}$ and a section $t=(t_1,t_2)\in\Model{k+1}(\C)$ (where $t_1\in\Model{k}(C_1)$ and $t_2\in\Model{k}(C_2)$) such that $s\in\flatten(t)$ and $\neg\CLC(\Model{k+1},t)$. In particular, this means that there exists a compatible family
\[
F=\{t_\K\in\FModel{k+1}(\K)\}_{\K\in\Cov{k+1}}
\]
such that $t_\C=t$. Given a context $C\in\Cov{k}$ we know that there exists a $C'\in\Cov{k}$ such that $\{C,C'\}\in\Cov{k+1}$. Let $u_C:=t_{\{C,C'\}}\mid_C$. This is well-defined because, given a different $C''\in\Cov{k}$ such that $\{C,C''\}\in\Cov{k+1}$, we have
\[
\begin{split}
t_{\{C,C''\}}\mid_C&\stackrel{(\ast)}{=}t_{\{C,C''\}}\mid_{\{C\}}=t_{\{C,C''\}}\mid_{\{C,C'\}\cap \{C,C''\}}\stackrel{(\dagger)}{=}t_{\{C,C'\}}\mid_{\{C,C'\}\cap \{C,C''\}}\\
&=t_{\{C,C'\}}\mid_{\{C\}}\stackrel{(\ast)}{=}t_{\{C,C'\}}\mid_C,
\end{split}
\]
where we have used compatibility of $F$ in $(\dagger)$, and applied what discussed in Remark \ref{rem: subtlety} in $(\ast)$. Thus we can define the family
\[
F':=\{u_C\in\FModel{k}(C)\}_{C\in\Cov{k}}.
\]
We can show that $F'$ is a compatible family for $\FModel{k}$ as follows: suppose $C,C'\in\Cov{k}$ and $C\cap C'\neq\emptyset$, then
\[
u_C\mid_{C\cap C'}=\left(t_{\{C,C'\}}\mid_C\right)\mid_{C\cap C'}=t_{\{C,C'\}}\mid_{C\cap C'}=\left(t_{\{C,C'\}}\mid_{C'}\right)\mid_{C\cap C'}=u_{C'}\mid_{C\cap C'}.
\]
Now, because $s\in\flatten(t)$, we can suppose w.l.o.g. that $s\in\flatten(t_1)$. Moreover, because $t_\C=t=(t_1,t_2)$, we have $u_{C_1}=t_1$. Thus $F'$ is a compatible family which contains $t_1$. We conclude that $\neg\CLC(\Model{k},t_1)$, which implies $\neg\CLCk{k}(\S,s)$, as $s\in\flatten(t_1)$. The same argument can be used to prove that $\CSC(\Model{k+1})\Rightarrow \CSC(\Model{k})$. 
\end{proof}

From these two propositions, we immediately have the following theorem:

\begin{theorem}
Let $\S$ be a model on a scenario $\scenario$, and suppose $\S$ has the CCP. Let $N$ denote the size of the largest cycle in $\Cov{1}$. We have 
\[
\LC(\S,s)\Leftrightarrow \CLCk{N-1}(\S,s),
\]
for all section $s$ of $\S$. Moreover, we have
\[
\SC(\S)\Leftrightarrow \CSC(\Model{N-1}).
\]
\end{theorem}

This result shows that, if a model has the CCP, then studying its contextuality is equivalent to study the cohomological contextuality of its $(N-1)$-th joint model. In other words, cohomological contextuality on $\Model{N-1}$ is a full invariant for contextuality on the original model. As for cyclic scenarios, note that it might be possible to erase cohomological false positives for a particular model even at a lower level.  

Obviously, we usually do not know a priori whether a model satisfies the CCP, however, as we mentioned earlier, this property is extremely common among empirical models, which means that this method is widely applicable. In the following section we will give some examples to support this claim.

\subsection{Examples}
\paragraph{A simple scenario}
Let us start with the model summarised in Table \ref{tab: CCP}.
\begin{table}[htbp]
\centering
\begin{tabular}{c | c c c c}
\hline
Contexts & $(0,0)$ & $(0,1)$ & $(1,0)$ & $(1,1)$\\
\hline
$\{a,b\}$ & $0$ & $1$ & $1$ & $0$\\
$\{a,d\}$ & $1$ & $0$ & $1$ & $1$\\
$\{b,c\}$ & $1$ & $1$ & $0$ & $1$\\
$\{b,d\}$ & $1$ & $0$ & $0$ & $1$\\
$\{c,d\}$ & $0$ & $1$ & $1$ & $0$\\
\end{tabular}
\caption{The empirical model $\S$.}\label{tab: CCP}
\end{table}
A bundle diagram representation of the model can be found in Figure \ref{fig: bundle CCP}. 
\begin{figure}[htbp]
\centering
\includegraphics[scale=0.55]{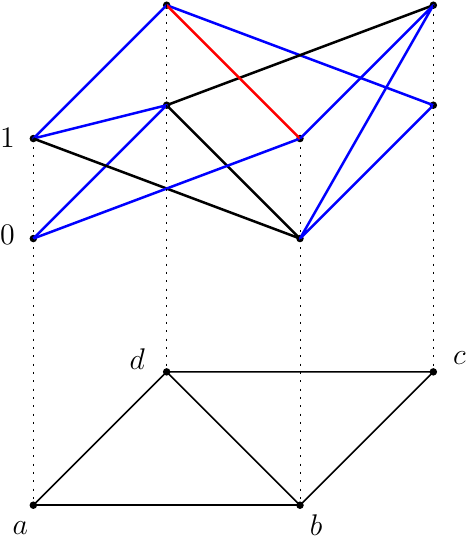}
\caption{The bundle diagram of the model summarised in Table \ref{tab: CCP}. In blue is highlighted the cohomological false positive for the section $(b,d)\mapsto(1,1)$, marked in red.}\label{fig: bundle CCP}
\end{figure}
By simply looking at the diagram, it is easy to see that the section $(b,d)\mapsto(1,1)$, marked in red, cannot be extended to any compatible family for $\S$. However, it can be extended to a compatible family for $\F$, as shown in blue. 

Using the enumeration specified in Table \ref{tab: enum CCP}, we can represent the first joint model $\Model{1}$ as a bundle diagram in Figure \ref{fig: CCPcohomology}. 

\begin{table}[htbp]
\centering
\begin{tabular}{c | c c c c}
\hline
Contexts & $(0,0)$ & $(0,1)$ & $(1,0)$ & $(1,1)$\\
\hline
$\{a,b\}$ &  & $s_1$ & $s_2$ & \\
$\{a,d\}$ & $s_3$ &  & $s_4$ & $s_5$ \\
$\{b,c\}$ & $s_6$ & $s_7$ & & $s_8$\\
$\{b,d\}$ & $s_9$ & & & $s_{10}$\\
$\{c,d\}$ &  & $s_{11}$ & $s_{12}$ & \\
\end{tabular}
\caption{An enumeration of the sections of $\S$. The model is logically contextual at $s_{10}$}\label{tab: enum CCP}
\end{table}

\begin{figure}[htbp]
\centering
\includegraphics[scale=0.55]{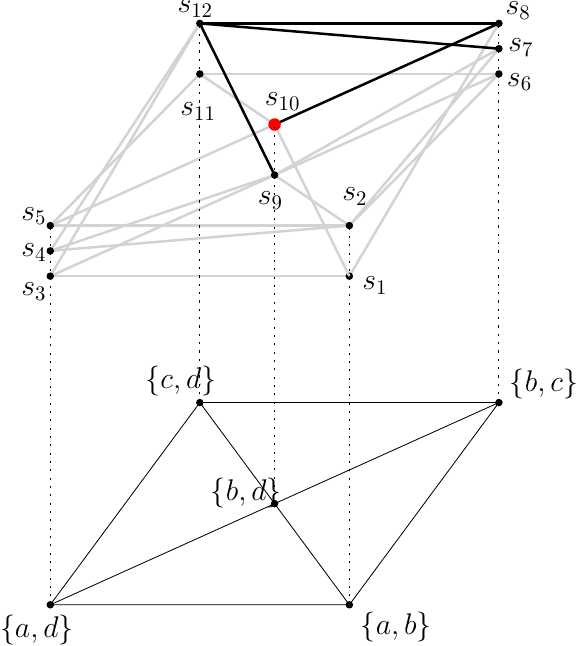}
\caption{The bundle diagram of $\Model{1}$. In black, all the possibilities to extend $s_{10}$ to a cohomology loop on the cycle $\{\{b,d\},\{b,c\},\{c,d\}\}$. They all fail to be compatible.}\label{fig: CCPcohomology}
\end{figure}

Notice that the section $s_{10}$, marked in red, cannot be extended to a compatible family for $\FModel{1}$ for the cycle $\{\{b,d\},\{b,c\},\{c,d\}\}$ (all the possibilities are highlighted in black). In particular, this means that the cohomological false positive has been deleted. Note that in this case it was sufficient to derive the first joint model to avoid a false positive. The size of the largest cycle in this scenario is $4$, thus, in general, we would have to consider the $3$rd joint model to remove any false positive with absolute certainty.

\paragraph{The Kochen-Specker model of \cite{Abramsky3}} The only cohomological false positive on a non-cyclic model that has appeared in the litterature is the Kochen-Specker model for the cover 
\begin{equation}\label{equ: KS}
\{A,B,C\},\{B,D,E\},\{C,D,E\},\{A,D,F\},\{A,E,G\}.
\end{equation}
It was introduced in \cite{Abramsky3} as an example of a cohomological false positive for a strongly contextual model. Let us show how the false positive arises. In Table \ref{tab: KS} we introduce a list of variables in $\mathbb{Z}_2$ for each of the $15$ possible sections of the model, to see if it is possible to construct a global section (i.e. a compatible family). 

\begin{table}[htbp]
\centering
\begin{tabular}{c | c c c c}
\hline
Contexts & $(1,0,0)$ & $(0,1,0)$ & $(1,0,0)$\\
\hline
$\{A,B,C\}$ & $a$ & $b$ & $c$\\
$\{B,D,E\}$ & $d$ & $e$ & $f$\\
$\{C,D,E\}$ & $g$ & $h$ & $i$\\ 
$\{A,D,F\}$ & $j$ & $k$ & $l$\\
$\{A,E,G\}$ & $m$ & $n$ & $o$\\ 
\end{tabular}
\caption{Variables for the possible sections of the Kochen-Specker model on the cover \eqref{equ: KS}.}\label{tab: KS}
\end{table}
The compatibility conditions of a presumed compatible family for $\F$ translate into equations modulo $2$. First of all, we have
\[
\begin{aligned}[c c c c]
a &= j = m \\
b &=d =g =c\\
e &= h = k \\
f &= i = n 
\end{aligned}
\]
Moreover,
\[
\begin{aligned}[c]
a\oplus c &= d\oplus f\\
a\oplus b &= h\oplus i\\
b\oplus c &=k\oplus l\\
\end{aligned}
\quad\quad
\begin{aligned}[c]
b\oplus c &= n\oplus o\\
d\oplus f &= j\oplus l\\
d\oplus e &= m\oplus o\\ 
\end{aligned}
\quad\quad
\begin{aligned}[c]
g\oplus i &= j\oplus l\\
g\oplus h &= m\oplus o\\
k\oplus l &= n\oplus o\\ 
\end{aligned}
\]
From these equations it follows that
\[
\begin{aligned}[c c c c]
a &=i=j=m=n=o\\
b &=c=d=e=g=h=k=l  \\
\end{aligned}
\]
Thus, we can rewrite Table \ref{tab: KS} to obtain Table \ref{tab: KS2}.

\begin{table}[htbp]
\centering
\begin{tabular}{c | c c c c}
\hline
Contexts & $(1,0,0)$ & $(0,1,0)$ & $(1,0,0)$\\
\hline
$\{A,B,C\}$ & $a$ & $b$ & $b$\\
$\{B,D,E\}$ & $b$ & $b$ & $a$\\
$\{C,D,E\}$ & $b$ & $b$ & $a$\\ 
$\{A,D,F\}$ & $a$ & $b$ & $b$\\
$\{A,E,G\}$ & $a$ & $a$ & $a$\\ 
\end{tabular}
\caption{Table \ref{tab: KS} rewritten given compatibility equations.}\label{tab: KS2}
\end{table}
Thanks to this table, we can immediately check that the model is strongly contextual. Indeed, in order to construct a compatible section for the model, we are only allowed to choose one section per context to which we assign $1$, while the others must be zero. By simply looking at Table \ref{tab: KS2} we can see that this is clearly impossible. 

However, if we let $a=1$ and $b=0$ we obtain the following cohomological compatible family
\[
\begin{split}
\{s_{\{A,B,C\},A} &,~ s_{\{B,D,E\},E},~ s_{\{C,D,E\},E},\\
&\phantom{,}~s_{\{A,D,F\},A},~ s_{\{A,E,G\},A}\oplus s_{\{A,E,G\},E}\oplus s_{\{A,E,G\},G}\}
\end{split}
\]
where we have used the standard notation for sections of a Kochen-Specker model, i.e. given a context $C$ and a measurement $m\in C$, the section $s_{C,m}$ is the section that maps $m$ to $1$ and every other $x\in C$ to $0$. 

This compatible family is a false positive for logical contextuality at sections $s_{\{A,B,C\},A}$, $s_{\{B,D,E\},E}$, $s_{\{C,D,E\},E}$ and $s_{\{A,D,F\},A}$. Furthermore, note that the only other families we have, namely the ones obtained by setting $a=0, b=1$ or $a=b=1$, do not give rise to false positives for any section since they contain multiple sections for every context. 

We will now show that it is sufficient to derive the first joint model $\Model{1}$ to remove all the cohomological false positives. 

First of all, we represent the first joint model using bundle diagrams. For each context $C=\{c_1,c_2,c_3\}$ of $\M$, there are exactly three possible sections, namely $s_{C,c_1}$, $s_{C,c_2}$ and $s_{C,c_3}$. Therefore, for each vertex of $\Cov{1}$, there are three distinct vertices in its fiber, which we will label with $s_{C,c_1}$, $s_{C,c_2}$ and $s_{C,c_3}$ from bottom to top (this labelling is not shown in the pictures for the sake of readability of the diagrams).
Using this convention, we have depicted the bundle diagram of the first joint model in Figure \ref{fig: FJMKS} (the colored vertices are only used as a visual reference).

\begin{figure}[htbp]
\centering
\includegraphics[scale=0.49]{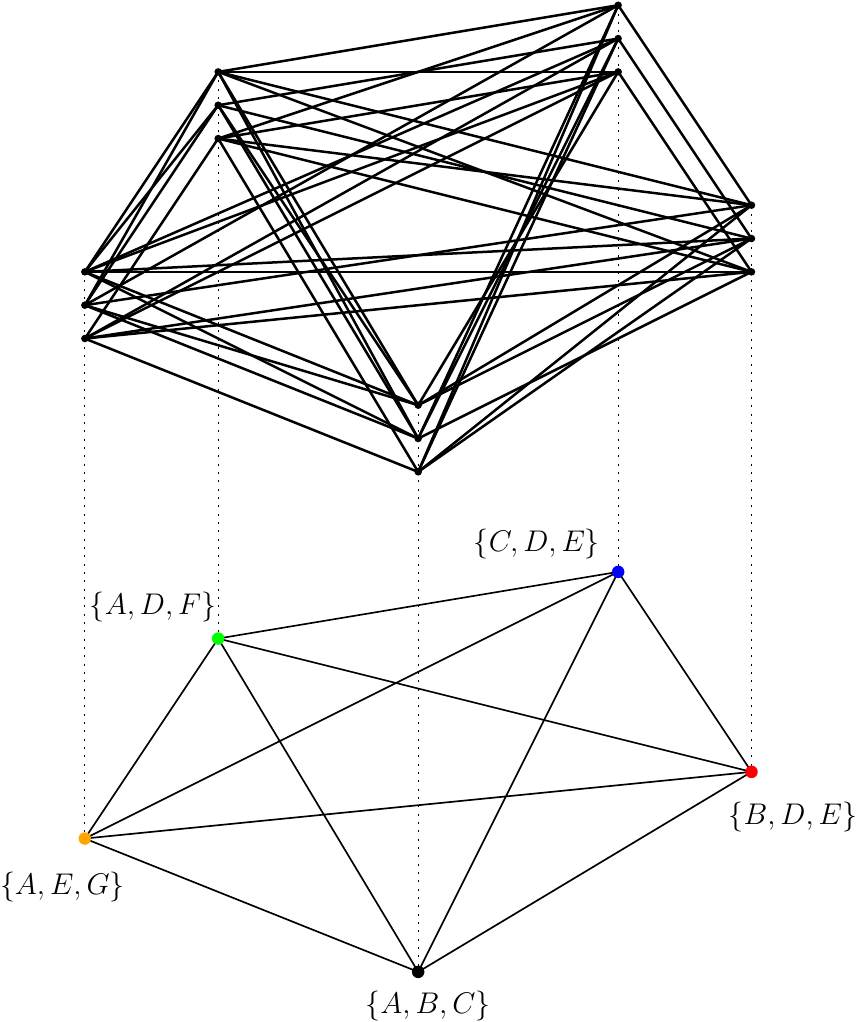}
\caption{The bundle diagram of the first joint model of the Kochen-Specker model on the cover \eqref{equ: KS}.}\label{fig: FJMKS}
\end{figure}

In Figure \ref{fig: KSPlanar}, we give a different representation of the model by decomposing it into two planar diagrams. The top diagram corresponds to the cycle that constitutes the perimeter of the pentagon, while the bottom one corresponds to the star-shaped cycle in the center. The colored circles in the fibers represent the four sections for which we have a cohomological false positive in the original model, as explained at the bottom of the picture. 

\begin{figure}[htbp]
\centering
\includegraphics[scale=0.49]{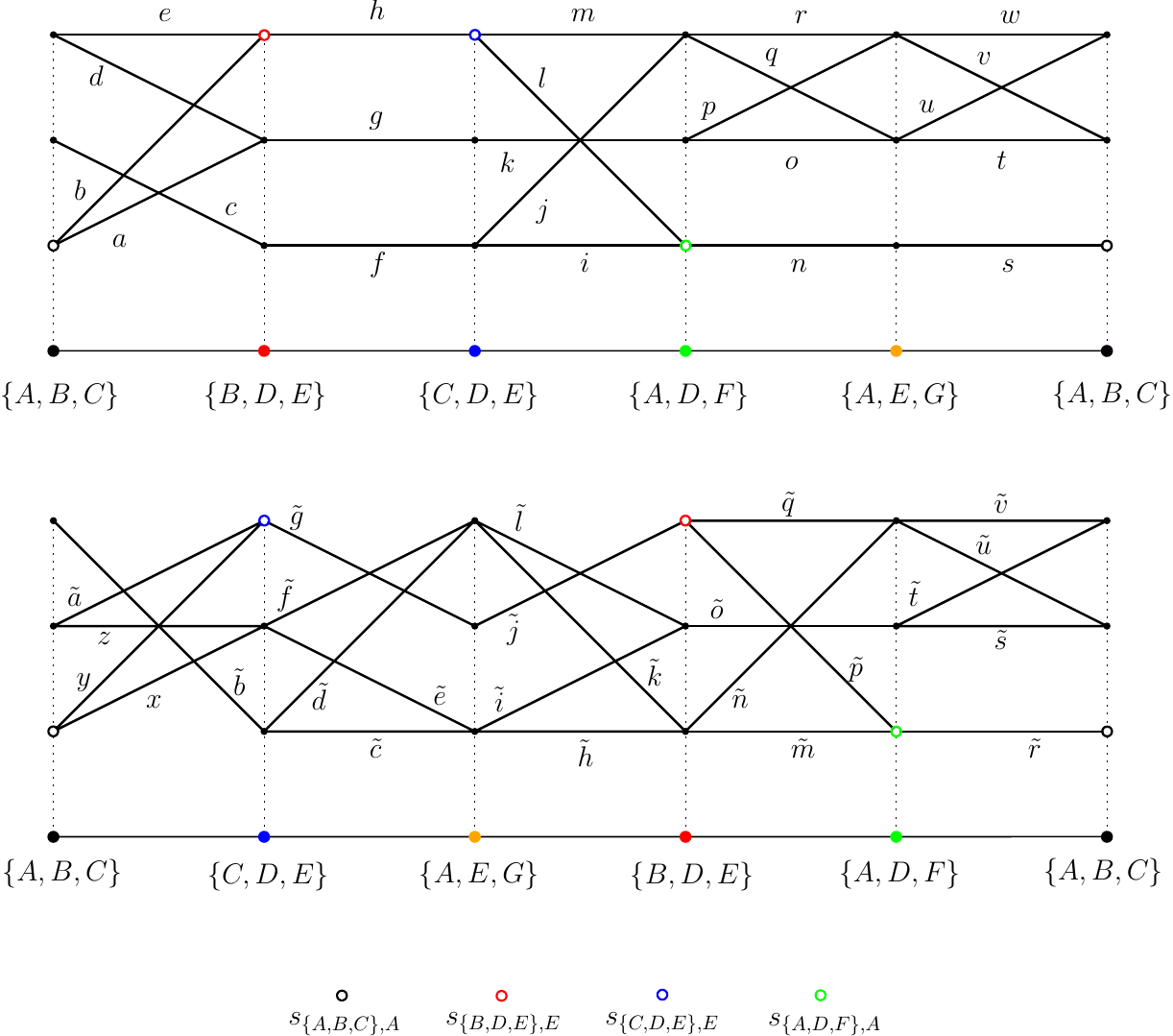}
\caption{The Kochen-Specker model of \cite{Abramsky3} decomposed in two cyclic planar diagrams. The top diagram corresponds to the perimeter of the pentagon; the bottom diagram refers to the central `star'. We introduce one variable in $\mathbb{Z}_2$ for each section of the model.}\label{fig: KSPlanar}
\end{figure}.

In the same picture, we have also introduced variables 
\[
a,b,\dots ,y,z, \tilde{a},\tilde{b},\dots, \tilde{u},\tilde{v}
\]
in $\mathbb{Z}_2$ for each of the possible sections of the first joint model. We will now show that the cohomological false positive no longer exists. To do so, we list all the equations imposed by compatibility conditions:
\[
\begin{aligned}[c]
a\oplus d &= g\\
b\oplus e &= h\\
c &= f\\
a\oplus d &= \tilde{o}\\
b\oplus e &= \tilde{p}\oplus\tilde{q}\\
c &= \tilde{m}\oplus\tilde{n}\\
a\oplus d &= \tilde{i}\oplus \tilde{l}\\
b\oplus e &= \tilde{j}\\
c &= \tilde{h}\oplus\tilde{k}\\
a\oplus b &= \tilde{r}\\
c &= \tilde{s}\oplus \tilde{u}\\
d\oplus e &= \tilde{t}\oplus\tilde{v}\\
a\oplus b &= s\\
c &= t\oplus v\\
d\oplus e &= u \oplus w\\
a\oplus b &= x\oplus y\\
c &= z\oplus\tilde{a}\\
d\oplus e &= \tilde{b}\\
f &= i\oplus j\\
g &= k\\
h &= l\oplus m\\
f &= b_1\\
g &= x\oplus z\\
h &= y \oplus\tilde{a}\\
f &= \tilde{h}\oplus\tilde{i}\\
g &= \tilde{j}\\
h &= \tilde{k}\oplus\tilde{l}\\
f &= \tilde{h}\oplus\tilde{k}\\
g &=\tilde{i}\oplus\tilde{l}\\
h &= \tilde{j}\\
\end{aligned}
\quad\quad
\begin{aligned}[c]
f &= \tilde{m}\oplus\tilde{n}\\
g &= \tilde{o}\\
h &=\tilde{p}\oplus\tilde{q}\\
i\oplus l &= n\\
k &= o\oplus p\\
j\oplus m &= q\oplus r\\
i\oplus l &= \tilde{r}\\
k &= \tilde{s}\oplus\tilde{t}\\
j\oplus m&=\tilde{u}\oplus\tilde{v}\\
i\oplus l &= \tilde{m}\oplus\tilde{p}\\
k &= \tilde{o}\\
j\oplus m &= \tilde{n}\oplus\tilde{q}\\
i\oplus j &= \tilde{b}\\
k &= x\oplus z\\
l\oplus m &= y\oplus\tilde{a}\\
i\oplus j &= \tilde{c}\oplus\tilde{d}\\
k &= \tilde{e}\oplus\tilde{f}\\
l\oplus m &= \tilde{g}\\
n &= s\\
o\oplus q &= t\oplus u\\
p\oplus r &= v\oplus w\\
n &= \tilde{h}\oplus\tilde{i}\\
o\oplus q &= \tilde{j}\\
p\oplus r &= \tilde{k}\oplus\tilde{l}\\
n &= \tilde{c}\oplus\tilde{l}\\
o\oplus q &= \tilde{g}\\
p\oplus r &= \tilde{d}\oplus\tilde{f}\\
n &= \tilde{m}\oplus\tilde{p}\\
o\oplus p &=\tilde{o}\\
q\oplus r &= \tilde{n}\oplus\tilde{q}\\
\end{aligned}
\quad\quad
\begin{aligned}[c]
n &=\tilde{r}\\
o\oplus r &=\tilde{s}\oplus\tilde{t}\\
q\oplus r &= \tilde{u}\oplus\tilde{v}\\
s &= x\oplus y\\
t\oplus v &= z\oplus \tilde{a}\\
u\oplus w &= \tilde{b}\\
s &= \tilde{r}\\
t\oplus v &=\tilde{s}\oplus\tilde{u}\\
u\oplus w &=\tilde{t}\oplus\tilde{v}\\
s &=\tilde{c}\oplus\tilde{e}\\
t\oplus w &= \tilde{g}\\
v\oplus w &=\tilde{d}\oplus\tilde{f}\\
s &= \tilde{h}\oplus\tilde{i}\\
t\oplus u &=\tilde{j}\\
v\oplus w &=\tilde{k}\oplus\tilde{l}\\
\tilde{b} &= \tilde{c}\oplus\tilde{d}\\
x\oplus z &=\tilde{e}\oplus\tilde{f}\\
y\oplus \tilde{a} &= \tilde{g}\\
x\oplus y &= \tilde{r}\\
z\oplus \tilde{a} &= \tilde{s}\oplus\tilde{u}\\
\tilde{b} &= \tilde{t}\oplus\tilde{v}\\
\tilde{c} \oplus\tilde{e} &=\tilde{h}\oplus\tilde{i}\\
\tilde{g} &=\tilde{j}\\
\tilde{d}\oplus\tilde{f} &=\tilde{k}\oplus\tilde{l}\\
\tilde{h}\oplus\tilde{k}&=\tilde{m}\oplus\tilde{n}\\
\tilde{i}\oplus\tilde{l} &= \tilde{o}\\
\tilde{j} &= \tilde{p}\oplus\tilde{q}\\
\tilde{m} \oplus \tilde{p} &=\tilde{r}\\
\tilde{o} &= \tilde{s}\oplus\tilde{t}\\
\tilde{n} \oplus\tilde{q} &= \tilde{u}\oplus\tilde{v}
\end{aligned}
\]

With the aid of a computer, we can easily find the solutions to this system of equations. The free variables are $a$, $b$, $i$, $o$, $\tilde{a}$, $\tilde{c}$, $\tilde{m}$, $\tilde{s}$ and $\tilde{t}$, and we must have
\begin{equation}\label{equ: solution}
\begin{split}
c &=f=g=h=k=n=s=\tilde{b}=\tilde{g}=\tilde{j}=\tilde{o}=\tilde{r}=a\oplus b\\
j &=l=a\oplus b\oplus i\\
p &= q = a\oplus b\oplus o\\
u &= v = a\oplus b \oplus t\\
y &= z = a\oplus b\oplus\tilde{a}\\
\tilde{d} &=\tilde{e}=\tilde{h}=\tilde{l}=a\oplus b\oplus\tilde{c}\\
\tilde{n} &=\tilde{p}=a\oplus b\oplus\tilde{m}\\
\tilde{t} &= \tilde{u}=a\oplus b \oplus \tilde{s}
\end{split}
\end{equation}
Consider section $s_{\{A,B,C\},A}$ of the original model (marked with a red circle in Figure \ref{fig: KSPlanar}). The only section of $\Model{1}$ at the context $\{\{A,B,C\}, \{A,E,G\}\}\in\Cov{1}$ that contains $s_{\{A,B,C\},A}$ is $s:=(s_{\{A,B,C\},A}, s_{\{A,E,G\}, A})$, whose corresponding variable is $s$. If we impose $s=1$ and $t=u=v=w=0$, we can see that these constraints are not consistent with the values \eqref{equ: solution} imposed by compatibility of a presumed compatible family for cohomology. Indeed, we have 
\[
0=u=a\oplus b\oplus t=s\oplus t=1\oplus 0=1.
\]
This means that the section $s=(s_{\{A,B,C\},A}, s_{\{A,E,G\}, A})$ cannot be extended to a compatible family for $\FModel{1}$. In other words, the cohomological false positive for $s_{\{A,B,C\},A}$ has vanished.


In the same way, $(s_{\{B,D,E\},E}, s_{\{C,D,E\},E})$ is the only section of $\Model{1}$ at the context $\{\{B,D,E\},\{C,D,E\}\}\in\Cov{1}$ that contains both $s_{\{B,D,E\},E}$ and $s_{\{C,D,E\},E})$. The corresponding variable is $h$, and if we impose $h=1$ and $g=f=0$, we have an immediate contradiction since $h=f=g$ by \eqref{equ: solution}. Thus we conclude, using the same argument as before, that the cohomological false positive for the contextuality of $\S$ at sections $s_{\{B,D,E\},E}$ and $s_{\{C,D,E\},E}$ has vanished.

Finally, to show that we have removed the false positive for $s_{\{A,D,F\},A}$, it is sufficient to argue that $(s_{\{A,D,F\},A},s_{\{A,E,G\},A})$ is the only section of $\Model{1}$ at the context $\{\{A,D,F\},\{A,E,G\}\}\in\Cov{1}$ that contains it. The corresponding variable is $n$, and if we impose $n=1$, $o=p=q=r=0$, we have
\[
0=p=a\oplus b\oplus o= n\oplus o=1\oplus 0=1,
\]
which is again a contradiction.

\section{Conclusions and further directions}
Thanks to the joint model construction, we introduced a cohomology obstruction to the extension of local sections, which represents a complete invariant for contextuality in the extremely vast class of models satisfying the cyclic contextuality property.
We showed how this invariant gets rid of all the known false positives from the literature, and proved its efficacy in general models. There are strong indications suggesting that this cohomology obstruction represents a full invariant for contextuality in all models. Indeed, in order to give rise to a false positive, a model $\S$ on a scenario $\scenario$ would have to satisfy all of the following:
\begin{itemize}
\item $\S$ is contextual.
\item $\scenario$ is non-cyclic.
\item $\S$ does not satisfy the CCP.
\item $\S$ gives rise to a cohomology false positive in \emph{all} its joint versions at \emph{all} its sections.
\end{itemize}
For this reason, we propose the following conjecture
\begin{conjecture}
Given a general model $\S$ on a scenario $\scenario$, there exists a $k\ge 0$ such that
\[
\LC(\S,s)\Leftrightarrow \CLCk{k}(\S,s),
\]
for all sections $s$ of $\S$. In other words, the cohomology of joint models represents a full invariant for contextuality. 
\end{conjecture}
Proving this conjecture will be our main focus in future work.

\section*{Acknowledgements} I would like to thank Samson Abramsky, Rui Soares Barbosa, and Shane Mansfield for helpful discussions. Support from the Oxford-Google DeepMind Graduate Scholarship and the EPSRC Doctoral Training Partnership is also gratefully acknowledged.

\bibliography{FullCohomology}{}
\bibliographystyle{alpha}

\end{document}